\documentclass[12pt,reqno,a4paper]{amsart}
\usepackage{amsmath,amssymb,graphicx,cite,enumerate, amsthm,subfig,dsfont}
\usepackage{mathpazo}
\usepackage[colorlinks=false]{hyperref}
\def\beq{\begin{equation}}
\def\eeq{\end{equation}}
\def\bea{\begin{eqnarray}}
\def\eea{\end{eqnarray}}
\def\wx{\widetilde x}

\def\ep{\varepsilon}
\def\tilde{\widetilde}

\def\epsilon{\varepsilon}

\newtheorem{theorem}{Theorem}
\newtheorem{fact}{Fact}

\makeatletter
\expandafter\let\expandafter
\reset@font\csname reset@font\endcsname
\def\subeqnarray{\arraycolsep1pt
    \def\@eqnnum\stepcounter##1{\stepcounter{subequation}
        {\reset@font\rm(\theequation\alph{subequation})}}
\jot5mm     \eqnarray}

\makeatother

\newcommand{\bbZ}{{\mathbb Z}}
\newcommand{\bbR}{{\mathbb R}}
\newcommand{\bbP}{{\mathbb P}}
\newcommand{\bbQ}{{\mathbb Q}}
\newcommand{\bbC}{{\mathbb C}}

\newcommand{\cO}{{\mathcal O}}

\newbox\meibox
\def\placeunder#1#2#3#4{\setbox\meibox%
\vbox{\hbox{\hskip#4$\hphantom{#2}$}\hbox{$\hphantom{#1}$}}%
\vtop{\baselineskip=0pt\lineskiplimit=\baselineskip%
\lineskip=#3\hbox to \wd\meibox{\hfil\hskip#4$#2$\hfil}%
\hbox to \wd\meibox{\hfil$#1$\hfil}}}

\begin{document}
%%%%%%%%%%%%%%%%%%%%%%%%%%%%%%%
%%%%%%%%%%%%%%%%%%%%%%%%%%%%%%%

\title[On elliptic solutions of integrable birational maps]
{On the construction of elliptic solutions of integrable birational maps}
\author[M. Petrera \and A. Pfadler \and Yu. Suris (with Appendix by Yu. Fedorov)]{Matteo Petrera \and Andreas Pfadler \and Yuri B. Suris \\ (with Appendix by Yuri N. Fedorov)}

\thanks{M. Petrera, A. Pfadler, Yu. Suris: Institut f\"ur Mathematik, MA 7-2,
Technische Universit\"at Berlin, Str. des 17. Juni 136, 10623 Berlin, Germany.}
\thanks{E-mail: {\tt  petrera@math.tu-berlin.de, andreas.pfadler@gmail.com, suris@math.tu-berlin.de}}

\maketitle

\begin{center}\today\end{center}

%%%%%%%%%%%%%%%%%%%%%%%%%%%%%%%
%%%%%%%%%%%%%%%%%%%%%%%%%%%%%%%
\begin{abstract}
We present a systematic technique to find explicit solutions of birational maps, provided that these solutions are given in terms of elliptic functions. The two main ingredients are: (i) application of classical addition theorems for elliptic functions, and (ii) experimental technique to detect an algebraic curve containing a given sequence of points in a plane. These methods are applied to Kahan-Hirota-Kimura discretizations of the periodic Volterra chains with 3 and 4 particles.

\end{abstract}
%%%%%%%%%%%%%%%%%%%%%%%%%%%%%%%
%%%%%%%%%%%%%%%%%%%%%%%%%%%%%%%

%%%%%%%%%%%%%%%%%%%%%%%%%%%%%%%
%%%%%%%%%%%%%%%%%%%%%%%%%%%%%%%
\section{Introduction}
%%%%%%%%%%%%%%%%%%%%%%%%%%%%%%%
%%%%%%%%%%%%%%%%%%%%%%%%%%%%%%%

This paper is devoted to finding solutions of integrable birational maps in terms of elliptic functions. The similar task for integrable ordinary differential equations was a popular and well-developed topic in the classical mathematics of 18-th and 19-th century. We mention here one of the most famous and prototypical examples, the Euler top. Differential equations of motion of the Euler top read
\begin{equation}\label{eq: ET x}
\left\{ \begin{array}{l}
\dot{x}_1=c_1 x_2 x_3, \vspace{.1truecm} \\
\dot{x}_2=c_2 x_3 x_1, \vspace{.1truecm} \\
\dot{x}_3=c_3 x_1x_2,
\end{array} \right.
\end{equation}
with real parameters $c_i$. They admit two functionally independent integrals of motion: one easily checks that the following three functions are integrals:
\begin{equation}\label{eq: ET H}
H_1=c_2x_3^2-c_3x_2^2,\qquad
H_2=c_3x_1^2-c_1x_3^2,\qquad
H_3=c_1x_2^2-c_2x_1^2,
\end{equation}
but only two of them are functionally independent because of $c_1H_1+c_2H_2+c_3H_3=0$. Using these integrals, one easily sees that the coordinates $x_j$ satisfy the following differential equations:
\begin{equation}\nonumber
%\label{eq: ET deg2}
\left\{ \begin{array}{l}
\dot{x}_1^2 = (H_3+c_2x_1^2)(c_3x_1^2-H_2),\vspace{.1truecm}\\
\dot{x}_2^2= (H_1+c_3x_2^2)(c_1x_2^2-H_3),\vspace{.1truecm}\\
\dot{x}_3^2 = (H_2+c_1x_3^2)(c_2x_3^2-H_1).
\end{array} \right.
\end{equation}
The fact that the polynomials on the right-hand sides of these equations are of degree four implies that the solutions are given by elliptic functions, as explained in the following classical theorem, see, e.g.,  \cite{Hal}.

\begin{theorem} \label{thm:Halphen}
The general solution of the scalar differential equation
\begin{equation}\label{eq: Halphen polynom}
\dot{x}^2=\alpha_0x^4+4\alpha_1 x^3 + 6\alpha_2x^2+4\alpha_3x+\alpha_4,\qquad
\alpha_i \in \mathbb{C},
\end{equation}
is given by the (time shifts of the) second order elliptic function
\begin{equation}\nonumber
x(t)=-\frac{\alpha_1}{\alpha_0} + \zeta(u+v)-\zeta(u)-\zeta(v)=-\frac{\alpha_1}{\alpha_0} + \frac{1}{2}\,\frac{\wp'(u)-\wp'(v)}{\wp(u)-\wp(v)},
\end{equation}
where $u=\sqrt{\alpha_0} \, t$,
while the point $v$ of the corresponding elliptic curve is determined by the relations
\begin{equation}\nonumber
\wp(v)=\frac{\alpha_1^2-\alpha_0 \alpha_2}{\alpha_0},\qquad \wp'(v)=\frac{\alpha_3 \alpha_0^2 - 3\alpha_0 \alpha_1 \alpha_2 + 2\alpha_1^3}{\alpha_0^3}.
\end{equation}
Here the invariants of the Weierstrass $\wp$-function are given by
\begin{equation}\nonumber
g_2=\frac{\alpha_0 \alpha_4 - 4\alpha_1\alpha_3 +3\alpha_2^2}{\alpha_0^2},
\qquad g_3=\frac{\alpha_0 \alpha_2\alpha_4 +2\alpha_1 \alpha_2 \alpha_3   - \alpha_2^3-\alpha_0 \alpha_3^2 - \alpha_1^2 \alpha_4} {\alpha_0^3}.
\end{equation}
\end{theorem}

We recall that the order of an elliptic function is, by definition, the number of its zeros (or poles) in an arbitrary parallelogram of periods, counted with multiplicities. While for a mathematician of the 19-th century it was quite a routine task to derive a differential equation of the type (\ref{eq: Halphen polynom}) for dependent variables of integrable systems, it seems that nothing comparable to that has been developed for discrete time integrable systems which only became popular more recently, after the advent of the modern theory of integrable systems (also known as the theory of solitons). In the few number of examples, where the exact solutions of integrable maps were actually computed, the following classical result going back to L. Euler was used. (One of the first known applications of this result to the theory of discrete integrable systems is due to R. Baxter \cite{Ba}.)

\begin{theorem}\label{thm: biquad}
Let $P(x,\wx)$ be an irreducible symmetric biquadratic polynomial over $\bbC$. Then the algebraic curve
$
\mathcal{C} = \bigl\{ (x,\wx) \in \bbC^2  :P(x,\wx) = 0\bigr\}
$
has genus 1 and may be parametrized as $(x,\wx)=(f(t),f(t+\delta))$,
with a second order elliptic function $f$ and some shift $\delta \in \bbC$. Conversely, for an elliptic function $f$ of order 2 and for an arbitrary $\delta \in \bbC$, $x=f(t)$ and $\wx = f(t+\delta)$ satisfy an algebraic relation of the form
$P(x,\wx) = 0$, where $P$ is an irreducible symmetric biquadratic polynomial.
\end{theorem}

Concrete computations can be performed as follows. To determine invariants of the elliptic curve described by a symmetric biquadratic relation
\begin{equation}\label{eq: biquad}
P(x,\wx) = \alpha_0 x^2 \widetilde{x}^2 + \alpha_1x \wx (x+\wx)
+\alpha_2(x^2+\widetilde{x}^2) + \alpha_3 x\widetilde{x}
+\alpha_4(x+\widetilde{x}) + \alpha_5=0,
\end{equation}
consider the Hamiltonian system of differential equations
$$
\dot x = {\displaystyle{\frac{\partial P(x,\wx)}{\partial \wx}}}, \qquad
\dot \wx = -{\displaystyle{\frac{\partial P(x,\wx)}{\partial x}}}.
$$
Using the relation $P(x,\wx)=0$, we can eliminate either $x$ or $\wx$ to get
\beq \label{eq:ddeode}
{\dot x}^2 = P(x), \qquad {\dot \wx}^2 = P(\wx),
\eeq
where
\bea
P(x) &=& (\alpha_1^2-4\alpha_0\alpha_2)x^4 + (2\alpha_1\alpha_3 -4\alpha_0 \alpha_4 -4\alpha_1\alpha_2)x^3 \nonumber  \\
& & +\, (\alpha_3^2-4 \alpha_0 \alpha_5 - 4\alpha_2^2-2\alpha_1\alpha_4)x^2 + (2\alpha_3 \alpha_4 - 4\alpha_1\alpha_5-4\alpha_2\alpha_4) x + \alpha_4^2 - 4\alpha_2 \alpha_5. \nonumber
\eea
Now, Theorem \ref{thm:Halphen} can be applied to differential equations (\ref{eq:ddeode}) in order to determine $g_2$ and $g_3$ in terms of $\alpha_i$.

To illustrate how a biquadratic relation of the type (\ref{eq: biquad}) can be derived for an integrable map, we briefly remind the relevant results for the so called Hirota-Kimura discretization of the Euler top \cite{HK, PS, PPS2}. The HK discretization of the Euler top (\ref{eq: ET x}) is:
\begin{equation}\label{eq: dET x}
\renewcommand{\arraystretch}{1.3}
\left\{\begin{array}{l}
\widetilde{x}_1-x_1=\epsilon c_1(\widetilde{x}_2x_3+x_2\widetilde{x}_3),\\
\widetilde{x}_2-x_2=\epsilon c_2(\widetilde{x}_3x_1+x_3\widetilde{x}_1),\\
\widetilde{x}_3-x_3=\epsilon c_3(\widetilde{x}_1x_2+x_1\widetilde{x}_2).
\end{array}\right.
\end{equation}
The map $f:x\mapsto\widetilde{x}\ $ is obtained by solving (\ref{eq: dET x}) for $\widetilde{x}$, and is a birational map of degree 6. As can be easily verified, this map possesses the following integrals of motion which are deformations of (\ref{eq: ET H}):
\begin{equation}\label{eq: dET ints}
    H_i(\epsilon)=\frac{c_j x_k^2-c_k x_j^2}{1-\epsilon^2 c_j c_k x_i^2}.
\end{equation}
Only two of them are independent, since
$$
c_1 H_1(\epsilon)+ c_2 H_2(\epsilon) + c_3 H_3(\epsilon)+
\epsilon^4 c_1 c_2 c_3 H_1(\epsilon) H_2(\epsilon) H_3(\epsilon)=0.
$$
To derive biquadratic relations for the components of the discrete time Euler top, one starts by proving the following relations:
\beq
\label{eq: dET W}
\left\{\begin{array}{l}
 \wx_2x_3-x_2\wx_3=\epsilon H_1(\epsilon)(\wx_1+x_1),\vspace{.1truecm}\\
 \wx_3x_1-x_3\wx_1=\epsilon H_2(\epsilon)(\wx_3+x_3),\vspace{.1truecm}\\
 \wx_1x_2-x_1\wx_2=\epsilon H_3(\epsilon)(\wx_3+x_3).
 \end{array} \right.
\eeq
From eqs. (\ref{eq: dET x}) and (\ref{eq: dET W}) there follows:
\[
(\wx_i-x_i)^2/(\epsilon c_i)^2+(\epsilon H_i(\epsilon))^2(\wx_i+x_i)^2=2(\wx_j^2x_k^2+x_j^2\wx_k^2).
\]
It remains to express $x_j^2$ and $x_k^2$ through $x_i^2$ and integrals $H_j(\epsilon)$, $H_k(\epsilon)$, which results in
\begin{equation}\nonumber
%\label{eq: dET biquad}
P_i(x_i,\wx_i) = \alpha_0^{(i)} x_i^2 \wx_i^2 + \alpha_2^{(i)}(x_i^2+\wx_i^2) + \alpha_3^{(i)}x_i\wx_i
+ \alpha_5^{(i)}=0,
\end{equation}
where the coefficients are expressed through the integrals of motion as follows:
\bea
%\label{eq: dET biquad coeffs}
& \alpha_0^{(i)}=-4\epsilon^2 c_j c_k, \quad
\alpha_2^{(i)}=(1+\epsilon^2 c_j H_j(\epsilon))(1-\epsilon^2 c_k H_k(\epsilon)), & \nonumber\\
& \alpha_3^{(i)}=-2(1-\epsilon^2 c_j H_j(\epsilon))(1+\epsilon^2 c_k H_k(\epsilon)), \quad
\alpha_5^{(i)}=4\epsilon^2 H_j(\epsilon) H_k(\epsilon). & \nonumber
\eea

Our goal in the present paper is to propose a toolbox which allows to compute the elliptic functions expressions for solutions of integrable maps going beyond the relatively simple situation covered by Theorem \ref{thm: biquad}. Our approach is based on some classical results about elliptic functions combined with rigorous symbolic computations in the spirit of experimental mathematics.

%%%%%%%%%%%%%%%%%%%%%%%%%%%%%%%
%%%%%%%%%%%%%%%%%%%%%%%%%%%%%%%
\section{Outline of the method}
\label{exper}
%%%%%%%%%%%%%%%%%%%%%%%%%%%%%%%
%%%%%%%%%%%%%%%%%%%%%%%%%%%%%%%

The theoretical basis of our method is the well known classical statement that any two elliptic functions with the same periods satisfy an algebraic relation. Specifically, we have the following claim \cite{WW}, which generalizes Theorem \ref{thm: biquad}.

\begin{theorem}\label{thm: WW}
Let  $X$ and $Y$ be two elliptic functions with the same periods, of orders $n$ and $m$, respectively. Then there exists an algebraic relation of the form
$
P(X,Y) = 0,
$
with an irreducible bivariate polynomial $P(X,Y)$ satisfying
\[
\deg_X P \leq m, \qquad \deg_Y P \leq n, \qquad \deg P \leq n+m.
\]
The coefficients of $P$ are unique up to multiplication with a scalar.
\end{theorem}

As a corollary, any elliptic function $X$ satisfies an algebraic differential equation, i.e., a polynomial relation between $X$ and its derivative $X'$. Similarly, if $X=X(t)$ is an elliptic function, then $\tilde{X}(t) = X(t+\delta)$ for an arbitrary $\ep\in \bbC$ is also an elliptic function with the same periods. Hence, $X$ and $\tilde{X}$ are connected by  a polynomial relation, called {\it{addition theorem}}. Theorems \ref{thm:Halphen}, \ref{thm: biquad} illustrate these conclusions. A famous concrete example of an addition theorem is given by the Weierstrass $\wp$-function. If we take the following well-known formula for $\wp$-functions,
$$
\wp(u+v) + \wp(u) + \wp(v) = \frac{1}{4} \left( \frac{\wp'(u) - \wp'(v)}{\wp(u) - \wp(v)} \right)^2,
$$
and eliminate all derivatives via $\wp'^2 = 4 \wp^3 - g_2 \wp - g_3$, we obtain
\beq \nonumber
%\label{eq:p_biquad}
\left(X Y+ YZ+ZX+\frac{g_2}{4}\right)^2 - 4(XYZ-g_3)(X+Y+Z) = 0,
\eeq
where
$
X = \wp(x), \, Y=\wp(y), \, Z=\wp(z),
$
such that $x+y+z=0$. Setting $z=\delta$, we obtain the symmetric biquadratic relation between $X=\wp(x)$ and $Y=\wp(x+\delta)$.
\medskip

Of course, the above relation may be of lower degree, as, for instance, for $X=\wp$ and $Y=\wp'$ (where $n=2$, $m=3$, but the total degree of $P(X,Y)=Y^2-4X^3+g_2X+g_3$ is only 3), or for an even more trivial example of $X=\wp$ and $Y=\wp^2$ (where $n=2$, $m=4$, and $P(X,Y)=Y-X^2$ is of bidegree (2,1)). This is a common phenomenon related to shared poles of $X$ and $Y$. If one imposes additional conditions on the poles of $X$ and $Y$, one may obtain sharper degree bounds for $P$. In particular, the following statement holds true.

\begin{theorem}\label{thm: sharp}
Let $X$ and $Y$ be two elliptic functions with the same periods, each one of order $n$ and having $n$ simple poles. If $X$ and $Y$ have $k$ poles in common, then $\deg P \leq 2n-k$.
\end{theorem}

Now we are in a position to sketch the essentials of our method. Assume that we are given a birational map $f$ on the phase space $\bbR^n$ with coordinates $x_i$,
$$
x \mapsto \tilde x = f(x).
$$
We want to test whether it is  solvable in terms of elliptic functions and to provide an Ansatz for explicit solutions.
If orbits of $f$ can be parametrized in terms of elliptic functions, then the pairs $(x_i,\wx_i)=(x_i(t),x_i(t+\delta))$  considered as functions of the discrete time $t \in \delta\bbZ$ will satisfy a polynomial relation of the type
\begin{equation}\label{eq: alg x wx}
P(x_i,\wx_i) = 0.
\end{equation}
The degree of $P$ in $x_i$ and $\wx_i$ will be the same and equal to the order of the elliptic function $x_i$.

Suppose that all components $x_i$ are elliptic functions. Then one could check whether they have the same periods. For this aim, one could compute the invariants $g_2$, $g_3$ of all the curves given by the above relations. The computation of the invariants $g_2$, $g_3$ may be accomplished using algorithms by M.~van Hoeij  \cite{VH}, implemented in Maple {\it algcurves} package.

Assume that all $x_i$ are given by elliptic functions with the same periods and (for the sake of simplicity) of the same order $n$. Then one could gather additional information about the distribution of their poles and zeros. For this aim, one looks for relations of the form
\begin{equation}\label{eq. alg xi xj}
Q(x_i,x_j) = 0.
\end{equation}
The degree of $Q$ now tells us whether $x_i$ and $x_j$ have a number of common poles. In particular, if $\deg Q = 2n-k$, then $x_i$ and $x_j$ are likely to have $k$ common poles. The same relation, considered as a polynomial relation between $1/x_i$ and $1/x_j$, will deliver information about possible common zeros of $x_i$ and $x_j$. Similarly, one looks for the number of coincidences between the poles of $x_i$ and zeros of $x_j$ (say). Eventually, we obtain enough information about the solutions in order to fully characterize poles and zeros of $x_i$. On the base of this information, one finally arrives at Ans\"atze, which can be rigorously verified using the classical means of mathematical analysis.
\medskip

In applications, formulas for the maps $f$ are often given by implicit equations of motion, like eq. (\ref{eq: dET x}), the explicit formulas being much more complicated and difficult to handle. Also, explicit formulas for integrals are often very messy (compare \cite{PPS1, PPS2}). Therefore, one cannot expect to be able to find the above mentioned relations (\ref{eq: alg x wx}), (\ref{eq. alg xi xj}) by hand, and one has to turn to the software for symbolic manipulations. We discuss some aspects of symbolic computations in the next section.

%%%%%%%%%%%%%%%%%%%%%%%%%%%%%%%
%%%%%%%%%%%%%%%%%%%%%%%%%%%%%%%
\section{Hirota-Kimura bases}
\label{sect: HK}
%%%%%%%%%%%%%%%%%%%%%%%%%%%%%%%
%%%%%%%%%%%%%%%%%%%%%%%%%%%%%%%

For our method, the notion of a Hirota-Kimura basis, introduced and studied in some detail in \cite{PPS1}, is relevant. We recall here the main facts, following that reference.

For a given birational map $f: \bbR^n \to \bbR^n$, a set of functions $\Phi=(\varphi_1,\ldots,\varphi_m)$, linearly independent over $\bbR$, is called a {\it Hirota-Kimura basis}
(HK-basis), if for every $x\in\bbR^n$ there exists a non-vanishing vector
$c=(c_1,\ldots,c_m)^{\rm T}$ such that 
\[
c_1\varphi_1(f^i(x))+\ldots+c_m\varphi_m(f^i(x))=0\quad{\rm for\;\; all} \quad i\in\bbZ.
\]
For a given $x\in\bbR^n$, the vector space consisting of all $c\in\bbR^m$ with
this property, say $K_\Phi(x)$, is called the {\it null-space} of the basis $\Phi$ at the point $x$.

The notion of HK-bases is closely related to the notion of integrals, even if they cannot be
immediately translated into one another. For instance, if $K_\Phi(x)={\rm span}(c_1,\ldots,c_M)^{\rm T}$, so that $\dim K_\Phi(x)=1$, then the quotients $c_i:c_j$ are integrals of motion (or, in other words, $[c_1:\ldots :c_m]\in\bbC\bbP^{m-1}$ is an integral of motion).

For a given set of functions $\Phi=(\varphi_1,\ldots,\varphi_m)$
and for any interval $[j,k]\subset\bbZ$ we denote
\begin{equation}\label{eq: balance}
X_{[j,k]}(x) = \left(
\begin{array}{ccc}
\varphi_1(f^j(x))  & ..   & \varphi_m(f^j(x))   \\
\varphi_1(f^{j+1}(x))  & ..   & \varphi_m(f^{j+1}(x))   \\
...  &   & ...   \\
\varphi_1(f^{k}(x))  & .. & \varphi_m(f^{k}(x))
\end{array}
\right).
\end{equation}
In particular, $X_{(-\infty,\infty)}(x)$ will denote the double
infinite matrix of the type (\ref{eq: balance}). Obviously,
\[
\ker X_{(-\infty,\infty)}(x)=K_\Phi(x).
\]
Thus, $\Phi$ is a Hirota-Kimura basis for $f$ if and only if $\dim \ker
X_{(-\infty,\infty)}(x)\ge 1$. Our algorithm for detecting this
situation is based on the following result formulated in \cite{PPS1}.
\begin{theorem}\label{Lem: dim K}
Let
\begin{equation}\label{eq: fundamental finite first}
\dim\ker X_{[0,s-1]}(x)= m-s \;\; {\rm for}\;\; 1 \leq s \leq m-1,
\end{equation}
and
\begin{equation}\label{eq: fundamental finite last}
\dim\ker X_{[0,m-1]}(x)=1
\end{equation}
hold for all $x\in\bbR^n$. Then for any $x\in\bbR^n$ there holds:
\begin{equation}\nonumber
K_\Phi(x)=\ker X_{[0,m-1]}(x),
\end{equation}
and, in particular,
\begin{equation}\nonumber
\dim K_\Phi(x)=1.
\end{equation}
\end{theorem}
Let us stress that, while in general (i.e., for general maps $f$, general function sets $\Phi$ and for generic  $x\in\bbR^n$) relations (\ref{eq: fundamental finite first}) are satisfied, one will generally find that the $m\times m$ matrix $X_{[0,m-1]}(x)$ is non-degenerate, so that $\dim\ker X_{[0,m-1]}(x)=\dim K_\Phi(x)=0$. Finding (a candidate for) a HK-basis $\Phi$ with $\dim\ker X_{[0,m-1]}(x)=\dim K_\Phi(x)=1$ is a highly non-trivial task.

The following numerical algorithm can be used for checking that a given system of functions $\Phi=(\varphi_1,\ldots,\varphi_m)$ is a HK-basis for a given birational map $f$.
\begin{itemize}
\item[(N)] For several randomly chosen initial points $x\in\bbQ^n$, compute $\ker X_{[0,m-1]}(x)$, using exact rational arithmetic. If for every $x$ condition (\ref{eq: fundamental finite last}) is satisfied, then, for all practical considerations, $\Phi$ is a HK-basis for $f$, with $\dim K_\Phi(x)=1$.
\end{itemize}

In what follows, we will call the results obtained in this way {\em Facts}. As such, they cannot be ascribed the status of mathematical theorems, but with an additional effort they can be given a rigorous proof (sometimes a human proof, and sometimes a computer assisted one). We have provided details for such proofs in \cite{PPS1} in the situation which was much more demanding than the ones described in the present paper. Therefore, we are pretty sure that all our {\em Facts} can be turned into mathematical theorems after a sufficient effort. However, this would not be strictly necessary for our purposes, because we use these {\em Facts} to formulate the Ans\"atze for solutions, which are then rigorously proved by different methods.
\smallskip

The relevance of the notion of HK-bases for our present goals is easy to understand. Indeed, existence of
an algebraic relation like (\ref{eq: alg x wx}) or (\ref{eq. alg xi xj}) for iterates of a birational map $f$
can be phrased as existence of a certain HK-basis. For instance, formula (\ref{eq: biquad}) for some component $x_i$ of the vector $x$ is equivalent to saying that
\[
\Phi=\Big(x_i^2\wx_i^2, \ x_i\wx_i(x_i+\wx_i), \ x_i^2+\wx_i^2, \ x_i\wx_i, \ x_i+\wx_i, \ 1\Big),
\]
is a HK basis for $f$ with $K_\Phi(x)$ spanned by $(\alpha_0,\ldots,\alpha_5)$. By the way, formula (\ref{eq: Halphen polynom}) can be interpreted as the claim that $$\Phi=(\dot{x}^2, x^4, x^3, x^2, x, 1)$$ is a HK-basis with $\dim K_\Phi(x)=1$ for the phase flow of the corresponding differential equation. Thus, actually  HK-bases appeared in a disguised form in the continuous time theory long ago.

In the remaining part of the paper we will apply the method described in the last two sections to finding explicit solutions, given in terms of elliptic functions, of two birational maps discretizing two
classical integrable systems, namely the periodic Volterra chains with $N=3$ and $N=4$.

%%%%%%%%%%%%%%%%%%%%%%%%%%%%%%%
%%%%%%%%%%%%%%%%%%%%%%%%%%%%%%%
\section{Elliptic solutions of the infinite Volterra chain}
%%%%%%%%%%%%%%%%%%%%%%%%%%%%%%%
%%%%%%%%%%%%%%%%%%%%%%%%%%%%%%%

The infinite Volterra chain (VC) is a well-known completely integrable system \cite{M, KvM} governed by the
equations of motion
\begin{equation}\label{eq: VL}
\dot{x}_n=x_n(x_{n+1}-x_{n-1}),\qquad n\in\mathbb Z.
\end{equation}
This system admits two different families of solutions which are expressed in terms of elliptic functions.

The first family of elliptic solutions \cite{Ves} is given by:
\begin{eqnarray}
x_n(t) & = & \zeta\big(t+nv\big)-\zeta\big(t+(n-1)v\big)+\zeta(v)-\zeta(2v)
\label{eq: VL sol1 1}\\
 & = & \frac{\sigma\big(t+(n+1)v\big)\sigma\big(t+(n-2)v\big)}
 {\sigma\big(t+nv\big)\sigma\big(t+(n-1)v\big)\sigma(2v)},\label{eq: VL sol1 2}
\end{eqnarray}
where $v$ is an arbitrary complex number.
The equivalence of the two representations (\ref{eq: VL sol1 1}) and (\ref{eq: VL sol1 2}) is either easily checked by looking at poles and zeroes of the both elliptic functions, or just by using formula
\begin{equation}\label{1234}
\zeta(a)+\zeta(b)+\zeta(c)-\zeta(a+b+c)=\frac{\sigma(a+b)\sigma(b+c)\sigma(c+a)}
{\sigma(a)\sigma(b)\sigma(c)\sigma(a+b+c)}.
\end{equation}
The check that (\ref{eq: VL sol1 1}) and (\ref{eq: VL sol1 2}) are indeed a solution of VC is now elementary. Take the logarithmic derivative of (\ref{eq: VL sol1 2}) and then use (\ref{eq: VL sol1 1}) with shifted indices:
\begin{eqnarray*}
\frac{\dot{x}_n}{x_n} & = & \zeta\big(t+(n+1)v\big)+\zeta\big(t+(n-2)v\big)-\zeta\big(t+nv\big)-\zeta\big(t+(n-1)v\big)\\
 & = & x_{n+1}-x_{n-1}.
\end{eqnarray*}

The second family of elliptic solutions \cite{Ver, Kit} is given by:
\begin{eqnarray}
\!\!\!\!\!\!\!\!\!\!x_{2n-1}(t) & = & \zeta\big(t+nv_1+(n-1)v_2\big)-\zeta\big(t+(n-1)(v_1+v_2)\big)+\zeta(v_2)-\zeta(v_1+v_2)\quad
\label{eq: VL sol2 1}
\\
 & = & \frac{\sigma\big(t+n(v_1+v_2)\big)\sigma\big(t+(n-1)v_1+(n-2)v_2\big)\sigma(v_1)}
 {\sigma\big(t+nv_1+(n-1)v_2\big)\sigma\big(t+(n-1)(v_1+v_2)\big)\sigma(v_2)\sigma(v_1+v_2)},
 \label{eq: VL sol2 2}
 \\ \nonumber\\
 x_{2n}(t) & = & \zeta\big(t+n(v_1+v_2)\big)-\zeta\big(t+nv_1+(n-1)v_2\big)+\zeta(v_1)-\zeta(v_1+v_2)
 \label{eq: VL sol2 3}
\\
 & = & \frac{\sigma\big(t+(n+1)v_1+nv_2\big)\sigma\big(t+(n-1)(v_1+v_2)\big)\sigma(v_2)}
 {\sigma\big(t+n(v_1+v_2)\big)\sigma\big(t+nv_1+(n-1)v_2\big)\sigma(v_1)\sigma(v_1+v_2)},
 \label{eq: VL sol2 4}
\end{eqnarray}
and reduces to the first one if $v_1=v_2=v$. Again, the verification of these solutions is straightforward:
\begin{eqnarray*}
\frac{\dot{x}_{2n-1}}{x_{2n-1}} & = & \zeta\big(t+n(v_1+v_2)\big)+\zeta\big(t+(n-1)v_1+(n-2)v_2\big)\nonumber\\
 & & -\zeta\big(t+nv_1+(n-1)v_2\big)-\zeta\big(t+(n-1)(v_1+v_2)\big)\\
 & = & x_{2n}-x_{2n-2},\\
 \frac{\dot{x}_{2n}}{x_{2n}} & = & \zeta\big(t+(n+1)v_1+nv_2)\big)+\zeta\big(t+(n-1)(v_1+v_2)\big)\nonumber\\
 & & -\zeta\big(t+n(v_1+v_2)\big)-\zeta\big(t+nv_1+(n-1)v_2\big)\\
 & = & x_{2n+1}-x_{2n-1}.
\end{eqnarray*}

The first family of solutions admits an $N$-periodic reduction $(n\in \mathbb Z/N\mathbb Z)$, if $Nv\equiv 0$ modulo the period lattice. The second family admits a $(2N)$-periodic reduction, if $N(v_1+v_2)\equiv 0$ modulo the period lattice. We will show that for the periodic VC with $N=3$ and $N=4$, these elliptic solutions are indeed general solutions.

%%%%%%%%%%%%%%%%%%%%%%%%%%%%%%%
%%%%%%%%%%%%%%%%%%%%%%%%%%%%%%%
\section{General solutions of the periodic Volterra chain with $N=3$}
%%%%%%%%%%%%%%%%%%%%%%%%%%%%%%%
%%%%%%%%%%%%%%%%%%%%%%%%%%%%%%%

The periodic reduction of the Volterra chain with $N=3$ (VC$_3$) reads:
\begin{equation}
\label{eq: VL3}
\left\{ \begin{array}{l}
\dot{x}_1= x_1(x_2-x_3),  \vspace{.1truecm} \\
\dot{x}_2= x_2(x_3-x_1), \vspace{.1truecm} \\
\dot{x}_3= x_3(x_1-x_2).
\end{array} \right.
\end{equation}
This integrable system possesses two functionally independent integrals of motion:
\begin{equation} \label{eq: VL3 H2}
H_1=x_1+x_2+x_3,\qquad H_2=x_1x_2x_3.
\end{equation}

We now prove that formulas (\ref{eq: VL sol1 1}) or (\ref{eq: VL sol1 2}) provide indeed the general
solution of VC$_3$.

\begin{theorem}
The general solution of VC$_3$ is given by formulas (\ref{eq: VL sol1 1}) or (\ref{eq: VL sol1 2}) with $v$ being a one third of a period, i.e., $3v\equiv 0$ modulo the period lattice.
Explicitly, in terms of $\zeta$-functions,
\begin{eqnarray}\label{eq: VL3 solution}
&&x_1  (t)= \zeta(t+v)-\zeta(t)+\zeta(v)-\zeta(2v),  \nonumber \\
&&x_2  (t)= \zeta(t+2v)-\zeta(t+v)+\zeta(v)-\zeta(2v), \nonumber \\
&&x_3  (t)= \zeta(t+3v)-\zeta(t+2v)+\zeta(v)-\zeta(2v), \nonumber
\end{eqnarray}
or, in terms of $\sigma$-functions,
\begin{eqnarray}
&&x_1  (t)=  \frac{\sigma(t-v)\sigma(t+2v)}
{\sigma(t)\sigma(t+v)\sigma(2v)},\label{eq: VL3 solution sigma1}\\
&&x_2 (t) = \frac{\sigma(t)\sigma(t+3v)}
{\sigma(t+v)\sigma(t+2v)\sigma(2v)},\label{eq: VL3 solution sigma2}\\
&&x_3  (t)= \frac{\sigma(t+v)\sigma(t+4v)}
{\sigma(t+2v)\sigma(t+3v)\sigma(2v)}.\label{eq: VL3 solution sigma3}
\end{eqnarray}
\end{theorem}

\begin{proof}
Eliminating $x_j,x_k$ from equations of motion (\ref{eq: VL3}) for $x_i$ with the help of integrals of motion (\ref{eq: VL3 H2}), one arrives at
\begin{equation}\label{eq: VL3 eq xi}
\dot{x}_i^2=x_i^2(x_i-H_1)^2-4H_2x_i.
\end{equation}
We can now apply Theorem \ref{thm:Halphen} to the differential equation (\ref{eq: VL3 eq xi}), which shows that the general solution is given by elliptic functions. The quartic polynomial in (\ref{eq: VL3 eq xi}) has the coefficients
$$
\alpha_0=1, \quad \alpha_1= -\frac{1}{2}H_1, \quad \alpha_2=\frac{1}{6}H_1^2, \quad \alpha_3=-H_2, \quad \alpha_4=0.
$$
Thus, we find that each of the coordinates $x_i$ is a time shift of
$$
x(t)=\zeta(t+v)-\zeta(t)-\zeta(v)+\frac{1}{2}H_1,
$$
where the zeta-functions correspond to the Weierstrass invariants
\begin{equation}\label{eq: VC3 g2g3 thru ints}
g_2=-2H_1H_2+\frac{1}{12}H_1^4,\quad g_3=-H_2^2+\frac{1}{6}H_1^3H_2-\frac{1}{216}H_1^6,
\end{equation}
and the shift $v$ is determined from
\begin{equation}\label{eq: VC3 v thru ints}
\wp(v)=\frac{1}{12}H_1^2,\quad \wp'(v)=-H_2.
\end{equation}
From \eqref{eq: VC3 g2g3 thru ints}, \eqref{eq: VC3 v thru ints} we find:
$$
H_1H_2= \frac{1}{24} H_1^4   -\frac{1}{2}g_2=6\wp^2(v)-\frac{1}{2}g_2=\wp''(v),
$$
which implies
$$
12\wp(v)(\wp'(v))^2=(\wp''(v))^2.
$$
This has to be compared with the duplication formula for the $\wp$-function,
$$
\wp(2v)=\frac{1}{4}\left(\frac{\wp''(v)}{\wp'(v)}\right)^2-2\wp(v).
$$
As a result, we find $\wp(2v)=\wp(v)$, so that $2v\equiv -v$, or $3v\equiv 0$. From the above formulas there follows:
$$
H_1=-\frac{\wp''(v)}{\wp'(v)}=4\zeta(v)-2\zeta(2v).
$$
%Note that setting $u_1=u_2=v$, $u_3=-2v\equiv v$ in the Frobenius-Stickelberger formula,
%$$
%\wp(u_1)+\wp(u_2)+\wp(u_3)=\big(\zeta(u_1)+\zeta(u_2)+\zeta(u_3)\big)^2,\qquad u_1+u_2+u_3=0,
%$$
%leads to $3\wp(v)=\big(2\zeta(v)-\zeta(2v)\big)^2$ for $3v\equiv 0$.
%
Finally, each of the coordinates $x_i$ is a time shift of
$$
x(t)=\zeta(t+v)-\zeta(t)-\zeta(v)+\frac{1}{2}H_1=\zeta(t+v)-\zeta(t)+\zeta(v)-\zeta(2v),
$$
We may eliminate any $x_k$ between equations (\ref{eq: VL3 H2}), getting $x_ix_j(x_i+x_j)-H_1x_ix_j+H_2=0$. Hence, any pair of functions $(x_i,x_j)$ satisfies a polynomial relation of degree $3$, which implies that any two functions $x_i$ and $x_j$ must have one common pole. Therefore, solutions of VC$_3$ must be as in (\ref{eq: VL3 solution sigma1})--(\ref{eq: VL3 solution sigma3}).
\end{proof}

%%%%%%%%%%%%%%%%%%%%%%%%%%%%%%%
%%%%%%%%%%%%%%%%%%%%%%%%%%%%%%%
\section{General solution of HK-discretization of Volterra chain with $N=3$}
%%%%%%%%%%%%%%%%%%%%%%%%%%%%%%%
%%%%%%%%%%%%%%%%%%%%%%%%%%%%%%%

We study here the discretization of the Volterra chain produced by so-called
{\it{Hirota-Kimura (HK) scheme}}. This scheme seems to be
introduced in the geometric integration literature by W. Kahan in
the unpublished notes \cite{K}. Kahan's discretization is applicable to any system of
ordinary differential equations with a
quadratic vector field, and automatically produces birational maps.
Probably unaware of the work by Kahan, this scheme was applied to integrable systems,
namely to the Euler top and to the Lagrange top, by R. Hirota and K. Kimura \cite{HK,KH}.
Later on, the authors of the present paper devoted a series of papers, \cite{HP,PPS1,PPS2,PS},
to integrability properties of this discretization technique,
providing a long list of integrable discretizations of systems belonging to the realm
of classical mechanics, including Volterra chains with $N=3$ and $N=4$. In the context of integrable systems, Kahan's discretization has been called
HK-discretization. This line of research was continued in \cite{C1,C2}.

The existence of two independent conserved quantities and of an invariant measure form of the HK-discretization of VC$_3$  (${\rm{d}}$VC$_3$) has been established in \cite{PPS2}. The discrete equations of motion  read
\begin{equation} \label{eq: dVL3}
\left\{ \begin{array}{l}
\wx_1-x_1=\epsilon x_1(\wx_2-\wx_3)+\epsilon\wx_1(x_2-x_3),  \vspace{.1truecm} \\
\wx_2-x_2=\epsilon x_2(\wx_3-\wx_1)+\epsilon\wx_2(x_3-x_1), \vspace{.1truecm} \\
\wx_3-x_3=\epsilon x_3(\wx_1-\wx_2)+\epsilon\wx_3(x_1-x_2).
\end{array} \right.
\end{equation}
The birational map generated by discrete equations of motion dVC$_3$,
possess the conserved quantities:
\begin{equation}\label{eq: dVL3 ints}
H_1=x_1+x_2+x_3, \qquad
H_2(\epsilon)=\frac{x_1x_2x_3}{1-\epsilon^2(x_1^2+x_2^2+x_3^2-2x_1x_2-2x_2x_3-2x_3x_1)}.
\end{equation}

To construct the general solution of dVC$_3$ in terms of elliptic functions, we will follow
the procedure described in section \ref{exper}. Recall that a {\em Fact} is an experimental
result obtained by a multiple run of Algorithm (N) with random initial data, which can be,
in principle, given a rigorous (maybe, computer assisted) proof.

\begin{fact}\label{fact: biquad curve_xi_wxi}
Along any orbit $\cO(x)=\{f^n(x)\}_{n\in \mathbb Z}$ of the map $f:x\mapsto\wx$, the pairs $(x_i,\wx_i)$ for all $i=1,2,3$ lie on a symmetric biquadratic curve, as in equation (\ref{eq: biquad}).
\end{fact}

According to Theorem \ref{thm: biquad}, we interpret Fact \ref{fact: biquad curve_xi_wxi} as follows.  The points of any orbit  $\cO(x)$ can be para\-metrized by the ``phase'' $t$ on an elliptic curve $\mathbb C/\Lambda$, where $\Lambda$ is the period lattice, so that the $f^n(x)$ corresponds to $t=n\delta+t_0$ with some $\delta\in\mathbb C$. Thus, the action of the map $f:x\mapsto\wx$ is represented by a shift $t\mapsto t+\delta$ on the elliptic curve $\mathbb C/\Lambda$. Moreover, each component $x_i$ as a function of $t$ is an elliptic function of order 2. Furthermore, the biquadratic curve $P(x_i,\wx_i)$ is one and the same for all three components $x_1, x_2, x_3$. Therefore all three components are obtained by shifting the argument of one and the same function. For symmetry reasons, we may assume that
$$
x_2(t)=x_1(t+v), \quad x_3(t)=x_2(t+v), \qquad 3v\equiv 0 \pmod \Lambda.
$$

\begin{fact}\label{fact: cubic curve_xi_xj}
Along any orbit $\cO(x)$, for each $i,j=1,2,3,$ the pairs $(x_i,x_j)$ lie on a symmetric biquadratic curve of total degree 3,
$$
Q_{ij}(x_i,x_j)=\alpha_1x_ix_j(x_i+x_j)+\alpha_2(x_i^2+x_j^2)+\alpha_3x_ix_j+\alpha_4(x_i+x_j)+\alpha_5=0.
$$
\end{fact}
As a consequence, according to Theorem \ref{thm: sharp}, every pair of functions $x_i$ and $x_j$ has one common pole. We may assume that the denominators of the functions $x_1$, $x_2$, $x_3$ are
$
\sigma(t)\sigma(t+v),\, \sigma(t+v)\sigma(t+2v), \, \sigma(t+2v)\sigma(t+3v),
$
respectively, just as in the solution of VC$_3$, see (\ref{eq: VL3 solution sigma1})--(\ref{eq: VL3 solution sigma3}).

\begin{fact}
Along any orbit $\cO(x)$, for each $i,j=1,2,3,$ the pairs $(x_i,1/\widetilde{x}_j)$ lie on a biquadratic curve of total degree 3.
\end{fact}

As a consequence, the functions $x_i$, $1/\widetilde{x}_j$ have a common pole. Therefore the two zeros of $x_i$ must be the $\delta$-shift and the $(-\delta)$-shift of the common pole of $x_j$ and $x_k$. We arrive at the conclusion that
\begin{eqnarray}
&&x_1(t) =\rho\,\frac{\sigma(t-v-\delta)\sigma(t+2v+\delta)}{\sigma(t)\sigma(t+v)}, \label{eq: dVC3 solution sigma1}\\
&&x_2 (t)= \rho\,\frac{\sigma(t-\delta)\sigma(t+3v+\delta)}{\sigma(t+v)\sigma(t+2v)},
\label{eq: dVC3 solution sigma2}\\
&&x_3(t) = \rho\,\frac{\sigma(t+v-\delta)\sigma(t+4v+\delta)}{\sigma(t+2v)\sigma(t+3v)}, \label{eq: dVC3 solution sigma3}
\end{eqnarray}
where the factor $\rho$ has to be determined. The other choice of the signs of the time shifts leads to the same functions, up to a constant factor.

The last missing portion of information necessary for a complete solution, namely, the expressions for the factors $\rho$ in the above formulas, is contained in Fact \ref{F2} below. The idea to consider the expressions mentioned there comes from the following equivalent form of equations of motion (\ref{eq: dVL3}):
\beq
\label{eq: dVC3 wx=x}
   \frac{\wx_i}{1+\epsilon(\wx_j-\wx_k)} = \frac{x_i}{1-\epsilon(x_j-x_k)},
\eeq
where $(i,j,k)$ is any cyclic permutation of $(1,2,3)$.
\begin{fact}\label{F2}
Along any orbit $\cO(x)$, for each cyclic permutation $(i,j,k)$ of the indices $(1,2,3)$, the pairs
$$
\left( \frac{x_i}{1\pm\epsilon(x_j-x_k)}, \frac{\wx_i}{1\pm\epsilon(\wx_j-\wx_k)} \right)
$$
lie on a symmetric biquadratic curve.
\end{fact}
As a consequence, the functions
\[
\frac{x_i}{1\pm\epsilon(x_j-x_k)}
\]
are elliptic functions of degree 2. In view of equations (\ref{eq: dVC3 wx=x}), Fact \ref{F2} yields also
that the two zeros of $x_1/\big(1-\epsilon(x_2-x_3)\big)$ must be $v-\delta, v$, while the two zeros of $x_1/\big(1+\epsilon(x_2-x_3)\big)$ must be $v, v+\delta$. In other words, the following relations must hold true:
\begin{equation}\label{eq: dVC3 solution aux}
1-\epsilon(x_2-x_3)\big|_{t=v+\delta}=0,\quad
1+\epsilon(x_2-x_3)\big|_{t=v-\delta}=0.
\end{equation}
Upon using formulas (\ref{eq: dVC3 solution sigma1})--(\ref{eq: dVC3 solution sigma3}) and taking into account that $3v\equiv 0$, both requirements in (\ref{eq: dVC3 solution aux}) result in one and the same formula for the factor $\rho$, namely,
\begin{equation}\label{eq: dVC3 rho aux1}
\frac{1}{\epsilon\rho}=\frac{\sigma(2v+2\delta)\sigma(v)}{\sigma(\delta)\sigma(v+\delta)}+
\frac{\sigma(2v)\sigma(v+2\delta)}{\sigma(v-\delta)\sigma(\delta)}.
\end{equation}
To simplify this expression, we observe that
$$
\sigma(2v+\delta)\sigma(v+\delta)\sigma(v)\sigma(2\delta)=
  \sigma(2v+2\delta)\sigma(v)\sigma(v-\delta)\sigma(\delta)
   + \sigma(2v)\sigma(v+2\delta)\sigma(v+\delta)\sigma(\delta), 
$$
which follows from the three-term functional equation for the $\sigma$-function,
\begin{eqnarray}\nonumber
&&\sigma(z+a)\sigma(z-a)\sigma(b+c)\sigma(b-c)+\sigma(z+b)\sigma(z-b)\sigma(c+a)\sigma(c-a)\nonumber\\
&&\qquad\qquad  \qquad +\,
\sigma(z+c)\sigma(z-c)\sigma(a+b)\sigma(a-b)=0, \label{eq: 3x4sigmas}
\end{eqnarray}
with the choice
\[
z=\frac{3v}{2}+\delta,\qquad a=\frac{v}{2}+\delta,\qquad b=\frac{v}{2}-\delta,\qquad
c=-\frac{v}{2}.
\]
Thus, we get
\begin{equation}\label{eq: dVC3 rho}
\frac{1}{\epsilon\rho}=\frac{\sigma(2v+\delta)\sigma(v)\sigma(2\delta)}{\sigma(v-\delta)\sigma^2(\delta)}.
\end{equation}

We arrive at the following statement, which can be now proven analytically.

\begin{theorem}
The general solution of dVC$_3$ is given by formulas (\ref{eq: dVC3 solution sigma1})--(\ref{eq: dVC3 solution sigma3}), with $3v\equiv 0$, where $\rho$ is given in (\ref{eq: dVC3 rho}). In terms of $\zeta$-functions,
\begin{eqnarray}
&&x_1 (t)= \tilde \rho \, (\zeta(t+v)-\zeta(t)+\zeta(v+\delta)-\zeta(2v+\delta)), \label{eq: dVC3 solution1}\\
&&x_2 (t)= \tilde \rho\, (\zeta(t+2v)-\zeta(t+v)+\zeta(v+\delta)-\zeta(2v+\delta)),\\
&&x_3(t) = \tilde \rho\, (\zeta(t+3v)-\zeta(t+2v)+\zeta(v+\delta)-\zeta(2v+\delta)),\label{eq: dVC3 solution3}
\end{eqnarray}
with
\begin{equation} \nonumber
%\label{eq: dVC3 rho1}
\frac{1}{\epsilon\tilde \rho}=
\frac{\sigma^2(v)\sigma(2\delta)}{\sigma(v+\delta)\sigma(v-\delta)\sigma^2(\delta)}
=2\zeta(\delta)-\zeta(v+\delta)+\zeta(v-\delta).
\end{equation}
Thus, for any initial point $x$, its iterations $f^n(x)$ are given by the above formulas with $t=n\delta+t_0$, with a suitable lattice of periods and suitable parameters $v,\delta,t_0\in\mathbb C$ (where, recall , $3v\equiv 0$).
\end{theorem}

\begin{proof}
We verify that formulas (\ref{eq: dVC3 solution sigma1})--(\ref{eq: dVC3 solution sigma3}), with $3v\equiv 0$, where $\rho$ is given in (\ref{eq: dVC3 rho}) provide indeed solutions of equations of motion (\ref{eq: dVC3 wx=x}). A similar computation can be done for solutions (\ref{eq: dVC3 solution1})--(\ref{eq: dVC3 solution3}) written in terms of $\zeta$-functions.

More concretely we want to prove that the following discrete equation of motion is satisfied:
$$
\frac{x_1}{1-\epsilon(x_2-x_3)}=\frac{\wx_1}{1+\epsilon(\wx_2-\wx_3)}.
$$
From (\ref{eq: dVC3 rho aux1}) and  (\ref{eq: dVC3 solution sigma2})--(\ref{eq: dVC3 solution sigma3}) we have:
\begin{eqnarray*}
\frac{1}{\epsilon\rho} \left(  1 - \ep (x_2-x_3) \right) & = &
\frac{\sigma(2v)\sigma(v+2\delta)}{\sigma(v-\delta)\sigma(\delta)}
+\frac{\sigma(2v+2\delta)\sigma(v)}{\sigma(v+\delta)\sigma(\delta)} \\
&& -\, \frac{\sigma(t-\delta)\sigma(t+3v+\delta)}{\sigma(t+v)\sigma(t+2v)}
+\frac{\sigma(t+v-\delta)\sigma(t+4v+\delta)}{\sigma(t+2v)\sigma(t+3v)}\\
\lefteqn{\hspace{-3.5cm} = \frac{\sigma(t+v)\sigma(t+2v)\sigma(2v)\sigma(v+2\delta)-
 \sigma(t-\delta)\sigma(t+3v+\delta)\sigma(v-\delta)\sigma(\delta)}
 {\sigma(t+v)\sigma(t+2v)\sigma(v-\delta)\sigma(\delta)}}\\
\lefteqn{\hspace{-3cm} +\, \frac{\sigma(t+2v)\sigma(t+3v)\sigma(2v+2\delta)\sigma(v)+
 \sigma(t+v-\delta)\sigma(t+4v+\delta)\sigma(v+\delta)\sigma(\delta)}
 {\sigma(t+2v)\sigma(t+3v)\sigma(v+\delta)\sigma(\delta)}.}
\end{eqnarray*}
Applying formula (\ref{eq: 3x4sigmas}) twice, first with
\[
z=t+\frac{3v}{2},\qquad a=\frac{v}{2},\qquad b=\frac{3v}{2}+\delta,\qquad c=\frac{v}{2}-\delta,
\]
and then with
\[
z=t+\frac{5v}{2},\qquad a=\frac{v}{2},\qquad b=\frac{3v}{2}+\delta,\qquad c=\frac{v}{2}+\delta,
\]
we obtain:
\begin{eqnarray}
\frac{1}{\epsilon\rho} \left(  1 - \ep (x_2-x_3) \right) & = & \frac{\sigma(t+2v-\delta)\sigma(t+v+\delta)\sigma(2v+\delta)\sigma(v+\delta)}
 {\sigma(t+v)\sigma(t+2v)\sigma(v-\delta)\sigma(\delta)}\nonumber\\
&   & +\, \frac{\sigma(t+3v+\delta)\sigma(t+2v-\delta)\sigma(2v+\delta)\sigma(v+\delta)}
 {\sigma(t+2v)\sigma(t+3v)\sigma(v+\delta)\sigma(\delta)} \nonumber\\
& = & \frac{\sigma(t+2v-\delta)\sigma(2v+\delta)\sigma(v+\delta)}{\sigma(t+2v)\sigma(\delta)} \nonumber \\
&   & \times
\left(\frac{\sigma(t+v+\delta)}{\sigma(v-\delta)\sigma(t+v)}+
\frac{\sigma(t+3v+\delta)}{\sigma(v+\delta)\sigma(t+3v)}\right). \nonumber
%\label{eq: dVC3 proof_l}
\end{eqnarray}
A similar computation gives
\bea
\frac{1}{\epsilon\rho} \left(  1 + \ep (\tilde x_2-\tilde x_3) \right) & = & \frac{\sigma(t+2v+2\delta)\sigma(2v+\delta)\sigma(v+\delta)}{\sigma(t+2v+\delta)\sigma(\delta)} \nonumber \\
&& \times
\left(\frac{\sigma(t+v)}{\sigma(v+\delta)\sigma(t+v+ \delta)}+
\frac{\sigma(t+3v)}{\sigma(v-\delta)\sigma(t+3v+\delta)}\right). \nonumber
%\label{eq: dVC3 proof_l}
\eea
Now a straightforward computation leads to
\bea
\frac{1-\epsilon(x_2-x_3)}{x_1}=  \ep \frac{\sigma(2v+\delta)}{\sigma(\delta) \sigma(v-\delta)}A(t,v,\delta) C_1 (t,v,\delta) , \nonumber \\
\frac{1+\epsilon(\tilde x_2-\tilde x_3)}{\tilde x_1}=  \ep \frac{\sigma(2v+\delta)}{\sigma(\delta) \sigma(v-\delta)}A(t,v,\delta)C_2 (t,v,\delta) , \nonumber
\eea
where
$$
A(t,v,\delta)=
\frac{\sigma(t+v+\delta)\sigma(t+3v)\sigma(v+\delta)+
\sigma(t+v)\sigma(t+3v+\delta)\sigma(v-\delta)}{\sigma(t+2v)\sigma(t+2v+\delta)},
$$
and
\bea
&&C_1(t,v,\delta) =
\frac{\sigma(t)\sigma(t+2v-\delta)}{\sigma(t-v-\delta)\sigma(t+3v)}, \nonumber\\
&&C_2(t,v,\delta) =\frac{\sigma(t+\delta)\sigma(t+2v)}{\sigma(t-v)\sigma(t+3v+\delta)}. \nonumber
\eea
The functions $C_1(t,v,\delta)$ and $C_2(t,v,\delta)$
are both constant (and equal to one another), since they are elliptic functions without zeros and poles, due to $3v\equiv 0$. The Theorem is thus proved.
\end{proof}

%It remains to express $v$ and the invariants $g_2$ and $g_3$ in terms of the integrals of motion.

%%%%%%%%%%%%%%%%%%%%%%%%%%%%%%%
%%%%%%%%%%%%%%%%%%%%%%%%%%%%%%%
\section{General solutions of the periodic Volterra chain with $N=4$}
%%%%%%%%%%%%%%%%%%%%%%%%%%%%%%%
%%%%%%%%%%%%%%%%%%%%%%%%%%%%%%%

The periodic reduction of the Volterra chain with $N=4$ (VC$_4$) reads:
\begin{equation} \nonumber
%\label{eq: VL4}
\left\{ \begin{array}{l}
\dot{x}_1= x_1(x_2-x_4),  \vspace{.1truecm} \\
\dot{x}_2= x_2(x_3-x_1), \vspace{.1truecm} \\
\dot{x}_3= x_3(x_4-x_2), \vspace{.1truecm} \\
\dot{x}_4= x_4(x_1-x_3).
\end{array} \right.
\eeq
This integrable system possesses three functionally independent integrals of motion:
\begin{equation} \label{eq: VL4 H2}
H_1=x_1+x_2+x_3+x_4,\qquad H_2=x_1x_3, \qquad H_3=x_2 x_4.
\end{equation}
We now prove that formulas (\ref{eq: VL sol2 1})--(\ref{eq: VL sol2 4}) provide indeed the general
solution of VC$_4$.

\begin{theorem}
The general solution of VC$_4$ is given by formulas (\ref{eq: VL sol2 1})--(\ref{eq: VL sol2 4}) with
$2(v_1+v_2) \equiv 0$.
Explicitly, in terms of $\zeta$-functions,
\begin{eqnarray}
&&x_1 (t)= \zeta(t+v_1) - \zeta(t) + \zeta(v_2) - \zeta(v_1+v_2), \nonumber \\
&&x_2  (t)= \zeta(t+v_1+v_2) - \zeta(t+v_1) + \zeta(v_1) - \zeta(v_1+v_2), \nonumber \\
&&x_3 (t) = \zeta(t+2v_1+v_2) - \zeta(t+v_1+v_2) + \zeta(v_2) - \zeta(v_1+v_2), \nonumber \\
&&x_4 (t) = \zeta(t+2v_1+2v_2) - \zeta(t+2v_1+v_2) + \zeta(v_1) - \zeta(v_1+v_2). \nonumber
\end{eqnarray}
In terms of $\sigma$-functions,
\begin{eqnarray}
&&x_1 (t) = \rho_1\frac{\sigma(t-v_2)\sigma(t+v_1+v_2)}{\sigma(t)\sigma(t+v_1)},
 \label{eq: VC4 fact 1} \\
&&x_2  (t)= \rho_2\frac{\sigma(t)\sigma(t+2v_1+v_2)}{\sigma(t+v_1)\sigma(t+v_1+v_2)},
 \label{eq: VC4 fact 2} \\
&&x_3 (t) = \rho_1\frac{\sigma(t+v_1)\sigma(t+2v_1+2v_2)}{\sigma(t+v_1+v_2)\sigma(t+2v_1+v_2)},
 \label{eq: VC4 fact 3} \\
&&x_4 (t) =\rho_2\frac{\sigma(t+v_1+v_2)\sigma(t+3v_1+2v_2)}{\sigma(t+2v_1+v_2)\sigma(t+2v_1+2v_2)},
 \label{eq: VC4 fact 4}
\end{eqnarray}
where
\begin{equation}\label{eq: VC4 rhos}
\rho_1=\frac{\sigma(v_1)}{\sigma(v_2)\sigma(v_1+v_2)},\qquad \rho_2=\frac{\sigma(v_2)}{\sigma(v_1)\sigma(v_1+v_2)}.
\eeq
\end{theorem}

\begin{proof}
One easily finds that $x_1,x_3$ satisfy the differential equation
\[
\dot{x}^2=(x^2-H_1x+H_2)^2-4H_3x^2,
\]
while $x_2,x_4$ satisfy a similar equation with $H_2\leftrightarrow H_3$. This is a differential equation of the form \eqref{eq: Halphen polynom} with
\[
\alpha_0=1, \quad \alpha_1=-\frac{1}{2}H_1, \quad \alpha_2 = \frac{1}{6} (H_1^2 +2 H_2-4 H_3), \quad \alpha_3 = -\frac{1}{2}H_1 H_2, \quad \alpha_4 = H_2^2.
\]
Theorem \ref{thm:Halphen} immediately leads to solution in terms of elliptic functions. The Weierstrass invariants can be expressed as
\begin{eqnarray*}
g_2  & = & \frac{1}{12} H_1^4 -\frac{2}{3} H_1^2 (H_2+H_3)+
\frac{4}{3} (H_2^2 +H_3^2- H_2 H_3),\\
g_3 & = & -{\frac 1 {216}} H_1^6 +\frac{1}{18} H_1^4 (H_2+H_3)-
\frac{1}{9} H_1^2 (2 H_2^2+2 H_3^2 +H_2 H_3)  \\
& & +\, {\frac 8 {27}} (H_2^3 +H_3^3) -
\frac{4}{9}H_2 H_3(H_2 +H_3).
\end{eqnarray*}
Note that $g_2$ and $g_3$ are symmetric with respect to the interchange $H_2 \leftrightarrow H_3$, which implies that all functions $x_i$ are elliptic functions with respect to the same period lattice.

We conclude that $x_1$ and $x_3$ are given by time shifts of the function
$$
x_{1,3}(t) = \zeta(t+v_1) - \zeta(t) - \zeta(v_1) + \frac{1}{2} H_1,
$$
where
$$
\wp(v_1)=\frac{1}{12}H_1^2-\frac{1}{3}H_2+\frac{2}{3}H_3, \qquad \wp'(v_1) =-H_3H_1.
$$
Similarly, $x_2$ and $x_4$ are time shifts of the function
$$
x_{2,4}(t) = \zeta(t+v_2) - \zeta(t) - \zeta(v_2) + \frac{1}{2} H_1,
$$
with
$$
\wp(v_2)=\frac{1}{12}H_1^2-\frac{1}{3}H_3+\frac{2}{3}H_2, \qquad \wp'(v_2)=-H_2H_1.
$$
With the help of the addition formula
$$
\wp(v_1)+\wp(v_2) + \wp(v_1+v_2) = \frac{1}{4}\left( \frac{\wp'(v_1)-\wp'(v_2)}{\wp(v_1)-\wp(v_2)}\right)^2,
$$
we find:
$$
\wp(v_1+v_2) = \frac{1}{12} H_1^2-\frac{1}{3}H_2-\frac{1}{3}H_3,
$$
which yields
$$
4 (\wp(v_1+v_2))^3-g_2\wp(v_1+v_2)-g_3 = 0.
$$
Thus, $\wp(v_1+v_2)$ is one of the roots of the Weierstrass cubic equation $4 z^3-g_2z-g_3=0$, which implies that
$v_1+v_2$ is a half period modulo the period lattice,  $2(v_1+v_2)\equiv0$.
%
%Now, from the above formulas it also follows that
%\[
% \zeta(v_1+v_2) - \zeta(v_1) - \zeta(v_2) = \frac{1}{2}\frac{ \wp'(v_1)-\wp'(v_2)}{\wp(v_1)-\wp(v_2)} = -\frac{1}{2} H_1.
%\]
%Finally, since $H_1$, $H_2$ and $H_3$ must be conserved quantities, $x_i$ must hence be of the form stated above.
\end{proof}

%%%%%%%%%%%%%%%%%%%%%%%%%%%%%%%
%%%%%%%%%%%%%%%%%%%%%%%%%%%%%%%
\section{General solution of HK-discretization of periodic Volterra chain with $N=4$}
%%%%%%%%%%%%%%%%%%%%%%%%%%%%%%%
%%%%%%%%%%%%%%%%%%%%%%%%%%%%%%%

The existence of three independent conserved quantities and of an invariant measure form of the HK-discretization of VC$_4$ (${\rm{d}}$VC$_4$) has been established in \cite{PPS2}. The discrete equations of motion read:
\begin{equation} \label{eq: dVL4}
\left\{ \begin{array}{l}
\wx_1-x_1=\epsilon x_1(\wx_2-\wx_4)+\epsilon\wx_1(x_2-x_4),  \vspace{.1truecm} \\
\wx_2-x_2=\epsilon x_2(\wx_3-\wx_1)+\epsilon\wx_2(x_3-x_1), \vspace{.1truecm} \\
\wx_3-x_3=\epsilon x_3(\wx_4-\wx_2)+\epsilon\wx_3(x_4-x_2),  \vspace{.1truecm} \\
\wx_4-x_4=\epsilon x_4(\wx_1-\wx_3)+\epsilon\wx_4(x_1-x_4).
\end{array} \right.
\end{equation}
The birational map generated by the discrete equations of motion dVC$_4$,
possess the conserved quantities:
$$
H_1=x_1+x_2+x_3+x_4, \qquad
H_2(\epsilon)=\frac{x_1x_3}{1-\epsilon^2(x_2-x_4)^2},\qquad
 H_3(\epsilon)=\frac{x_2x_4}{1-\epsilon^2(x_1-x_3)^2}.
$$

We are interested in constructing explicit elliptic solutions of dVC$_4$. To do so we will
proceed as in the case of dVC$_3$, starting by collecting preliminary experimental results
about the distribution of the points $(x_i,\wx_i)$ and $(x_i,x_j)$ on some curves.

\begin{fact} \label{JH}
Along any orbit $\cO(x)$, for each $i=1,2,3,4,$ the pairs $(x_i,\wx_i)$ lie on a curve of bidegree (4,4) without constant and linear terms. Equivalently, the pairs $(1/x_i,1/\wx_i)$ lie on a curve of bidegree (4,4) and of total degree 6. The corresponding curves coincide for $i=1,3$ and for $i=2,4$. They all are of genus 1.
\end{fact}

As a consequence,  $x_i$ as functions of $t$ are elliptic functions of degree 4. Moreover, $x_i$ and $\wx_i$ have two common zeros.

\begin{fact}
Along any orbit $\cO(x)$, the pairs $(x_i,x_j)$ lie on a curve of degree 4 if $i,j$ are of different parity,
and on a curve of degree 2 if $i,j$ are of the same parity.
\end{fact}

As a consequence, $x_1$ and $x_3$ are time shifts of one and the same function, and the same for $x_2$ and $x_4$.

\begin{fact} \label{HH}
Along any orbit $\cO(x)$, the pairs $(x_1+x_3,\wx_1+\wx_3)$ lie on a biquadratic curve, and the same holds true for the pairs $(x_2+x_4,\wx_2+\wx_4)$.
\end{fact}

As a consequence, functions $x_1+x_3$ and $x_2+x_4$ are of degree 2. Thus, the time shift relating $x_1$ and $x_3$ should be a half-period, and the same for $x_2$ and $x_4$. Therefore, we  assume
\begin{equation}
\label{eq: dVC4 shifts}
    x_3(t)=x_1(t+v_1+v_2),\quad x_4(t)=x_2(t+v_1+v_2),\quad 2(v_1+v_2)\equiv 0.
\end{equation}
We denote the common poles of $x_i$ by $0$, $-v_1$, $-(v_1+v_2)$, $-(2v_1+v_2)$.

With the help of all this information, we can proceed as follows. Denote the zeros of $x_1$ by $-a,-(a-\delta),-b,-(b-\delta)$, and the zeros of $x_2$ by $-c,-(c-\delta),-d,-(d-\delta)$.
Thus, we can finally write down the factorized expressions for $x_1$, $x_2$:
\bea
&&\!\!\!\!\!\!
x_1(t) = \rho_1\frac{\sigma(t+a)\sigma(t+a+2v_1+2v_2-\delta)\sigma(t+b)\sigma(t+b-\delta)}
{\sigma(t)\sigma(t+v_1)\sigma(t+v_1+v_2)\sigma(t+2v_1+v_2)},\label{m1}\\
&&\!\!\!\!\!\!
x_2 (t)= \rho_2\frac{\sigma(t+c)\sigma(t+c+2v_1+2v_2-\delta)\sigma(t+d)\sigma(t+d-\delta)}
{\sigma(t+v_1) \sigma(t+v_1+v_2) \sigma(t+2v_1+v_2) \sigma(t+2v_1+2v_2) }, \label{m2}
\eea
where $\rho_1$ and $\rho_2$ are two factors to be determined.
This choice of the factorization is justified by the continuous limit, $\delta\approx2\epsilon\to 0$, of (\ref{m1})--(\ref{m2})
which has to be compared with (\ref{eq: VC4 fact 1}) and (\ref{eq: VC4 fact 3}). Such a limit
tells us that
\begin{equation}  \nonumber
% \label{eq:_dVC4_ab}
a\approx -v_2, \qquad b\approx v_1+v_2, \qquad b+a-\delta=v_1,
\end{equation}
\begin{equation} \nonumber
% \label{eq:_dVC4_cd}
c\approx 0, \qquad d\approx 2v_1+v_2, \qquad c+d-\delta=2v_1+v_2.
\end{equation}
Eliminating $b$ and $d$ from the above expressions, we find:
\begin{eqnarray}
&&\!\!\!\!\!\!
x_1(t) = \rho_1\frac{\sigma(t+a) \sigma(t-a+v_1+\delta) \sigma(t-a+v_1) \sigma(t+a+2v_1+2v_2-\delta)}  {\sigma(t) \sigma(t+v_1) \sigma(t+v_1+v_2) \sigma(t+2v_1+v_2)},
\nonumber
% \label{eq: dVC4 ans 1}
\\
&&\!\!\!\!\!\!
x_2 (t) = \rho_2\frac {\sigma(t+c) \sigma(t-c+2v_1+v_2) \sigma(t-c+2v_1+v_2+\delta) \sigma(t+c+2v_1+2v_2-\delta)}{\sigma(t+v_1) \sigma(t+v_1+v_2) \sigma(t+2v_1+v_2) \sigma(t+2v_1+2v_2)}.
\nonumber
% \label{eq: dVC4 ans 2}
\end{eqnarray}
Factorized expressions for $x_3$, $x_4$ follow from (\ref{eq: dVC4 shifts}). However, it turns out to be convenient to have expressions for these variables with the same denominators as for $x_1$, $x_2$, respectively. This is achieved by using the quasi-periodicity of the $\sigma$-function with respect to the period $2(v_1+v_2) \equiv 0$. We get:
\begin{eqnarray}
&&\!\!\!\!\!\!\!\!\!\!
x_3 (t)= \rho_1\frac{\sigma(t-a-v_2+\delta) \sigma(t+a+v_1+v_2) \sigma(t+a+v_1+v_2-\delta) \sigma(t-a+2 v_1+v_2)}{\sigma(t) \sigma(t+v_1) \sigma(t+v_1+v_2) \sigma(t+2v_1+v_2)},
% \label{eq: dVC4 ans 3}
\nonumber\\
&&\!\!\!\!\!\!\!\!\!\!
x_4 (t)= \rho_2\frac{\sigma(t-c+v_1+\delta)\sigma(t+c+v_1+v_2) \sigma(t+c+v_1+v_2-\delta) \sigma(t-c+3v_1+2v_2)}{\sigma(t+v_1) \sigma(t+v_1+v_2) \sigma(t+2v_1+v_2) \sigma(t+2v_1+2v_2)}. \nonumber
%\label{eq: dVC4 ans 4}.
\end{eqnarray}
Next, we have to find the remaining unknowns $a$ and $c$, as well as $\rho_1$ and $\rho_2$. The idea to consider expressions from the crucial Fact \ref{fact dVC4 fractions} formulated below, comes from the observation that equations of motion (\ref{eq: dVL4}) can be equivalently re-written as
\begin{eqnarray}
\frac{\wx_1}{1+\epsilon(\wx_2-\wx_4)} & = & \frac{x_1}{1-\epsilon(x_2-x_4)},
 \label{eq: dVC4 alt 1}
\\
\frac{\wx_2}{1+\epsilon(\wx_3-\wx_1)} & = & \frac{x_2}{1-\epsilon(x_3-x_1)},
 \label{eq: dVC4 alt 2}
\\
\frac{\wx_3}{1+\epsilon(\wx_4-\wx_2)} & = & \frac{x_3}{1-\epsilon(x_4-x_2)},
 \label{eq: dVC4 alt 3}
\\
\frac{\wx_4}{1+\epsilon(\wx_1-\wx_3)} & = & \frac{x_2}{1-\epsilon(x_1-x_3)}.
 \label{eq: dVC4 alt 4}
\end{eqnarray}

\begin{fact}\label{fact dVC4 fractions}
Along any orbit $\cO(x)$, for each $i=1,2,3,4$, the pairs
\[
\left(\frac{x_i}{1\pm \epsilon(x_j-x_k) },\frac{\wx_i}{1 \pm \epsilon(\wx_j-\wx_k)}\right),
\]
where $j=i+1\pmod 4$, $k=i-1\pmod 4$,  lie on a symmetric biquadratic curve.
\end{fact}

As a consequence, the combinations
\[
\frac{x_i}{1\pm \epsilon(x_j-x_k) }
\]
are elliptic functions of degree 2. The structure of zeros of these functions can be inferred from equations of motion (\ref{eq: dVC4 alt 1})--(\ref{eq: dVC4 alt 4}). For instance, from (\ref{eq: dVC4 alt 1}) we see that the two zeros of $x_1/(1-\epsilon(x_2-x_4))$ are $-a,-b$, while the two zeros of $x_1/(1+\epsilon(x_2-x_4))$ are $-(a-\delta),-(b-\delta)$. Therefore,
\begin{equation}\label{eq: dVC4 solution aux_1}
1-\epsilon(x_2-x_4)\big|_{t=-a+\delta}=0,\qquad
1-\epsilon(x_2-x_4)\big|_{t=-b+\delta}=0,
\end{equation}
\begin{equation}\label{eq: dVC4 solution aux_2}
1+\epsilon(x_2-x_4)\big|_{t=-a}=0,\qquad
1+\epsilon(x_2-x_4)\big|_{t=-b}=0.
\end{equation}
Similarly, from (\ref{eq: dVC4 alt 2}) there follows that the two zeros of $x_2/(1-\epsilon(x_3-x_1))$ are $-c,-d$, while the two zeros of $x_2/(1+\epsilon(x_3-x_1))$ are $-(c-\delta),-(d-\delta)$, so that
\begin{equation}\label{eq: dVC4 solution aux_3}
1-\epsilon(x_3-x_1)\big|_{t=-c+\delta}=0,\qquad
1-\epsilon(x_3-x_1)\big|_{t=-d+\delta}=0,
\end{equation}
\begin{equation}\label{eq: dVC4 solution aux_4}
1+\epsilon(x_3-x_1)\big|_{t=-c}=0,\qquad
1+\epsilon(x_3-x_1)\big|_{t=-d}=0.
\end{equation}
From (\ref{eq: dVC4 alt 3}) we deduce that the two zeros of $x_3/(1-\epsilon(x_4-x_2))$ are $-a+v_1+v_2,-b+v_1+v_2$, while the two zeros of $x_3/(1+\epsilon(x_4-x_2))$ are $-(a-\delta)+v_1+v_2,-(b-\delta)+v_1+v_2$, so that
\begin{equation}\label{eq: dVC4 solution aux_5}
1-\epsilon(x_4-x_2)\big|_{t=-a+\delta+v_1+v_2}=0,\qquad
1-\epsilon(x_4-x_2)\big|_{t=-b+\delta+v_1+v_2}=0,
\end{equation}
\begin{equation}\label{eq: dVC4 solution aux_6}
1+\epsilon(x_4-x_2)\big|_{t=-a+v_1+v_2}=0,\qquad
1+\epsilon(x_4-x_2)\big|_{t=-b+v_1+v_2}=0.
\end{equation}
Finally, from (\ref{eq: dVC4 alt 4}) we conclude that the two zeros of $x_4/(1-\epsilon(x_1-x_3))$ are $-c+v_1+v_2,-d+v_1+v_2$, while the two zeros of $x_4/(1+\epsilon(x_1-x_3))$ are $-(c-\delta)+v_1+v_2,-(d-\delta)+v_1+v_2$, so that
\begin{equation}\label{eq: dVC4 solution aux_7}
1-\epsilon(x_1-x_3)\big|_{t=-c+\delta+v_1+v_2}=0,\qquad
1-\epsilon(x_1-x_3)\big|_{t=-d+\delta+v_1+v_2}=0,
\end{equation}
\begin{equation}\label{eq: dVC4 solution aux_8}
1+\epsilon(x_1-x_3)\big|_{t=-c+v_1+v_2}=0,\qquad
1+\epsilon(x_1-x_3)\big|_{t=-d+v_1+v_2}=0.
\end{equation}

Let us first concentrate on  (\ref{eq: dVC4 solution aux_3})--(\ref{eq: dVC4 solution aux_4})
and (\ref{eq: dVC4 solution aux_7})--(\ref{eq: dVC4 solution aux_8}). They result in eight conditions for $a$, $c$ and $\rho_1$. We show that actually almost all these conditions are equivalent, so that we are actually left with one condition for $c$ and one expression for $\rho_1$ through $c$ and $a$. For this aim, we first apply the tree-term formula (\ref{eq: 3x4sigmas}) with
$$
z=t+\frac{v_1}{2},\qquad a=a-\frac{v_1}{2},\qquad b=\frac{v_1}{2}+v_2+a-\delta,\qquad c=t+\frac{3v_1}{2}+v_2.
$$
to obtain the following expression:
\bea
\frac{1}{\rho_1} (x_1-x_3) & = & \frac{\sigma(t+a) \sigma(t-a+v_1) \sigma(t-a+v_1+\delta) \sigma(t+a+2v_1+2v_2-\delta)}{\sigma(t) \sigma(t+v_1) \sigma(t+v_1+v_2) \sigma(t+2v_1+v_2)}\nonumber \\
& & \hspace{-2cm}-\frac{\sigma(t-a-v_2+\delta) \sigma(t+a+v_1+v_2-\delta) \sigma(t+a+v_1+v_2) \sigma(t-a+2v_1+v_2)}{\sigma(t) \sigma(t+v_1) \sigma(t+v_1+v_2) \sigma(t+2v_1+v_2)} \nonumber \\
& =& -\frac{\sigma(2t+2v_1+v_2) \sigma(v_1+v_2) \sigma(v_1+v_2-\delta) \sigma(2a+v_2-\delta) } {\sigma(t) \sigma(t+v_1) \sigma(t+v_1+v_2) \sigma(t+2v_1+v_2) }. \nonumber
\eea
This function changes its sign by the shift $t\mapsto t+v_1+v_2$. Therefore conditions (\ref{eq: dVC4 solution aux_7})--(\ref{eq: dVC4 solution aux_8}) are equivalent to (\ref{eq: dVC4 solution aux_3})--(\ref{eq: dVC4 solution aux_4}). Furthermore, the above function changes it sign by $t\mapsto -t-2v_1-v_2$, therefore conditions (\ref{eq: dVC4 solution aux_3}) and (\ref{eq: dVC4 solution aux_4}) are equivalent. Thus, we can consider the first conditions in each of (\ref{eq: dVC4 solution aux_3})--(\ref{eq: dVC4 solution aux_4}) only. They result in two different values for $\rho_1$:
\begin{eqnarray}
\epsilon\rho_1 & = &- \frac{\sigma(c) \sigma(-c+v_1) \sigma(-c+v_1+v_2) \sigma(-c+2v_1+v_2)}
 {\sigma(-2c+2v_1+v_2) \sigma(v_1+v_2) \sigma(v_1+v_2-\delta) \sigma(2a+v_2-\delta)}
% \label{eq: vc4rho1_1}
\nonumber \\
 & = & -\frac{\sigma(-c+\delta) \sigma(-c+v_1+\delta) \sigma(-c+v_1+v_2+\delta) \sigma(-c+2v_1+v_2+\delta)}
{\sigma(-2c+2v_1+v_2+2\delta) \sigma(v_1+v_2) \sigma(v_1+v_2-\delta) \sigma(2a+v_2-\delta)},
\nonumber
\end{eqnarray}
which are equivalent if and only if the following condition holds:
\beq \label{eq: vc4_condition_c}
\frac {\sigma(c) \sigma(v_1-c) \sigma(v_1+v_2-c) \sigma(2 v_1+v_2-c) \sigma(2v_1+v_2+2\delta-2c)}
{\sigma(\delta-c) \sigma(v_1+\delta-c) \sigma(v_1+v_2+\delta-c) \sigma(2v_1+v_2+\delta-c)\sigma(2v_1+v_2-2c)} = 1.
\eeq
Equation (\ref{eq: vc4_condition_c}) determines $c$.

A similar computation can be performed for the function $(x_2-x_4)/\rho_2$. It turns first out that
conditions (\ref{eq: dVC4 solution aux_5})--(\ref{eq: dVC4 solution aux_6}) are equivalent to (\ref{eq: dVC4 solution aux_1})--(\ref{eq: dVC4 solution aux_2}). Then we find that (\ref{eq: dVC4 solution aux_1}) is
indeed equivalent to (\ref{eq: dVC4 solution aux_2}). We get two different values for $\rho_2$:
\begin{eqnarray}
\epsilon\rho_2 & = & \frac{\sigma(a+v_2) \sigma(a) \sigma(-a+v_1)\sigma(-a+v_1+v_2)}
{\sigma(-2a+v_1) \sigma(v_1+v_2) \sigma(v_1+v_2-\delta) \sigma(2c-v_1-\delta) }
% \label{eq: vc4rho2_1}
\nonumber \\
 & = & -\frac{\sigma(-a-v_2+\delta) \sigma(-a+\delta) \sigma(-a+v_1+\delta)\sigma(-a+v_1+v_2+\delta)}
{\sigma(-2a+v_1+2\delta) \sigma(v_1+v_2) \sigma(v_1+v_2-\delta) \sigma(2c-v_1-\delta) },
 \label{eq: vc4rho2_2} \nonumber
\end{eqnarray}
which are equivalent if and only if the following condition holds:
\beq \label{eq: vc4_condition_a}
\frac{ \sigma(a)\sigma(a+v_2) \sigma(a-v_1)\sigma(-a+v_1+v_2)\sigma(-2a+v_1+2\delta)}
{\sigma(\delta-a-v_2)\sigma(-a+\delta)\sigma(-a+v_1+\delta)\sigma(-a+v_1+v_2+\delta)
\sigma(-2a+v_1)}=1.
\eeq
Equation (\ref{eq: vc4_condition_a}) determines $a$.
It is easy to see that equation (\ref{eq: vc4_condition_c}) for $c$ and equation (\ref{eq: vc4_condition_a}) for $a+v_2$ are obtained from one another by the flip $v_1\leftrightarrow v_2$ (as they should).

We are now ready to prove the desired result.

\begin{theorem}
The general solution of $dVC_4$ is given by
\begin{eqnarray*}
&&\!\!\!\!\!\!\!\!\!\!\!x_1 (t)= \rho_1\frac{\sigma(t+a) \sigma(t-a+v_1+\delta) \sigma(t-a+v_1) \sigma(t+a+2v_1+2v_2-\delta)}  {\sigma(t) \sigma(t+v_1) \sigma(t+v_1+v_2) \sigma(t+2v_1+v_2)}, \\
&&\!\!\!\!\!\!\!\!\!\!\!x_2 (t)= \rho_2\frac {\sigma(t+c) \sigma(t-c+2v_1+v_2) \sigma(t-c+2v_1+v_2+\delta) \sigma(t+c+2v_1+2v_2-\delta)}{\sigma(t+v_1) \sigma(t+v_1+v_2) \sigma(t+2v_1+v_2) \sigma(t+2v_1+2v_2)}, \\
&&\!\!\!\!\!\!\!\!\!\!\!x_3 (t)= \rho_1\frac{\sigma(t-a-v_2+\delta) \sigma(t+a+v_1+v_2) \sigma(t+a+v_1+v_2-\delta) \sigma(t-a+2 v_1+v_2)}{\sigma(t) \sigma(t+v_1) \sigma(t+v_1+v_2) \sigma(t+2v_1+v_2)} \\
&&\!\!\!\!\!\!\!\!\!\!\!x_4 (t)=\rho_2\frac{\sigma(t-c+v_1+\delta)\sigma(t+c+v_1+v_2) \sigma(t+c+v_1+v_2-\delta) \sigma(t-c+3v_1+2v_2)}{\sigma(t+v_1) \sigma(t+v_{1 }+v_2) \sigma(t+2 v_{1 }+v_{2 }) \sigma(t+2 v_{1 }+2 v_{2 }) },
\end{eqnarray*}
where $2(v_1+v_2)\equiv 0$ and
\begin{eqnarray*}
\rho_1 &=& \frac{\sigma(c) \sigma(c-v_1) \sigma(-c+v_1+v_2) \sigma(-c+2v_1+v_2)}
 {\epsilon \sigma(-2c+2v_1+v_2) \sigma(v_1+v_2) \sigma(v_1+v_2-\delta) \sigma(2a+v_2-\delta)}, \\
\rho_2 &=&  \frac{\sigma(a)\sigma(a+v_2)  \sigma(-a+v_1)\sigma(-a+v_1+v_2)}
{\epsilon \sigma(-2a+v_1) \sigma(v_1+v_2) \sigma(v_1+v_2-\delta) \sigma(2c-v_1-\delta) }.
\end{eqnarray*}
Thus, for any initial point $x$, its iterations $f^n(x)$ are given by the above formulas with $t=n\delta+t_0$, with a suitable lattice of periods and suitable parameters $v_1,v_2,\delta,t_0\in\mathbb C$ (where, recall , $2(v_1+v_2)\equiv 0$).
The constants $a$ and $c$ are defined by (\ref{eq: vc4_condition_c})--(\ref{eq: vc4_condition_a}). 
\end{theorem}

\begin{proof}
We show how to verify (\ref{eq: dVC4 alt 2}). The remaining three equations may be dealt with in  the same way. Under conditions  (\ref{eq: vc4_condition_c})--(\ref{eq: vc4_condition_a})  the function $1+ \epsilon(x_3-x_1)$ has the zeros $-c$, $-d$, $-c-v_1-v_2+\delta$, $-d-v_1-v_2+\delta$, while the function $1- \epsilon(x_3-x_1)$ has the zeros $-c+\delta$, $-d+\delta$,$-c-v_1-v_2$,  $-d-v_1-v_2$. Hence, with the help of the periodicity condition $2(v_1+v_2)\equiv0$, it is easy to see that there holds
\begin{eqnarray} \label{eq: C1_1}
1&+&\epsilon(x_3-x_1) =    \\
 &   =& C_1{\frac { \sigma(t+c) \sigma(t-c+\delta+2 v_{1 }+v_{2 }) \sigma(t+c-\delta+v_{1 }+v_{2 }) \sigma(t-c+3 v_{1 }+2 v_{2 }) }{\sigma(t+v_{1 }) \sigma(t+v_{1 }+v_{2 }) \sigma(t+2 v_{1 }+v_{2 }) \sigma(t+2 v_{1 }+2 v_{2 }) }}, \nonumber
\end{eqnarray}
as well as
\begin{eqnarray} \label{eq: C2_1}
1&-&\epsilon(x_3-x_1) = \\
&= &C_2 {\frac {\sigma(t+c+v_{1 }+v_{2 }) \sigma(t-c+\delta+3 v_{1 }+2 v_{2 }) \sigma(t+c-\delta)
\sigma(t-c+2 v_1+v_2)}{\sigma(t+v_1) \sigma(t+v_1+v_2) \sigma(t+2 v_1+v_2) \sigma(t+2v_1+2 v_2) }}, \nonumber
\end{eqnarray}
with some constants $C_1$, $C_2$ depending on $v_1,v_2,c, \ep$. With the help of the three-term identity for the $\sigma$-function
(\ref{eq: 3x4sigmas}) we see that the difference $x_3-x_1$ has a zero at $t=-v_1-v_2/2$.  We therefore determine $C_1$ and $C_2$ by setting $t=-v_1-v_2/2$ in (\ref{eq: C1_1}) and (\ref{eq: C2_1}), giving
\begin{eqnarray*}
&&\!\!\!\!\!\!\!\! C_1=\frac{\sigma^2(v_2/2) \sigma(v_1+v_2/2) \sigma(v_1+3v_2/2)}
{\sigma(v_1+v_2/2-c) \sigma(v_1+v_2/2-c+\delta) \sigma(v_2/2+c-\delta) \sigma(2v_1+3v_2/2-c)},\\
&&\!\!\!\!\!\!\!\! C_2=\frac{\sigma^2(v_2/2) \sigma(v_1+v_2/2) \sigma(v_1+3v_2/2)}
{\sigma(v_2/2+c) \sigma(2v_1+3v_2/2+\delta-c) \sigma(v_1+v_2/2-c+\delta) \sigma(v_1+v_2/2-c)}.
\end{eqnarray*}
Hence,
\begin{equation}
\frac{C_1}{C_2} =  \frac{\sigma(v_2/2+c) \sigma(2v_1+3v_2/2+\delta-c)}
                        {\sigma(v_2/2+c-\delta) \sigma(2v_1+3v_2/2-c)}. \label{eq: ratioC}
\end{equation}

Now, by virtue of (\ref{eq: C1_1})--(\ref{eq: C2_1}), equation of motion (\ref{eq: dVC4 alt 2}) takes the form
\begin{eqnarray*}
\lefteqn{ \frac{\sigma(t+c+2v_1+2v_2) \sigma(t-c+\delta+2v_1+v_2)}
               {\sigma(t+c+v_1+v_2) \sigma(t+\delta-c+3v_1+2v_2)} } \nonumber \\
&& =
\frac{C_1}{C_2} \frac{\sigma(t+c) \sigma(t-c+\delta+2v_1+v_2) \sigma(t+c+2v_1+2v_2-\delta)}
                      {\sigma(t+c+v_1+v_2) \sigma(t+\delta-c+3v_1+2v_2) \sigma(t+c-\delta)}, \nonumber
\end{eqnarray*}
which reduces to
$${\frac{C_1}{C_2}} =
\frac{\sigma(t+c+2v_1+2v_2) \sigma(t+c-\delta)}{\sigma(t+c) \sigma(t+c+2v_1+2v_2-\delta)},
$$
but this identity is easily verified using (\ref{eq: ratioC}) and the quasi-periodicity of the $\sigma$-function.
\end{proof}

%%%%%%%%%%%%%%%%%%%%%%%%%%%%%%%
%%%%%%%%%%%%%%%%%%%%%%%%%%%%%%%
\section{Conclusions}
%%%%%%%%%%%%%%%%%%%%%%%%%%%%%%%
%%%%%%%%%%%%%%%%%%%%%%%%%%%%%%%

We have shown how to use Hirota-Kimura bases to construct explicit solutions of birational maps which are solvable in terms of elliptic functions. The appealing features of this approach are the following:
\begin{itemize}
\item It is systematic: Ans\"atze for explicit solutions are derived via searching for algebraic curves on which different two-dimensional projections of the orbits lie.  \vspace{.2truecm}
\item We do not look for or try to construct additional integrable structures (for instance Lax pairs), a process which would usually require large amounts of guesswork and/or research experience.
\end{itemize}
It would be interesting to apply this approach to further birational maps which can be solved in terms of elliptic functions. These are, basically, almost all Hirota-Kimura discretizations presented in \cite{PPS2}, with the exception of the discrete Clebsch system (which is likely to admit solutions in terms of theta functions of genus 2 Riemann surfaces), and include, among others, the discrete Lagrange top and the discrete Kirchhoff system.

\section*{Acknowledgments}
This research was supported by the DFG Collaborative Research Center TRR 109 ``Discretization in Geometry and Dynamics''. We would like to thank the referee for valuable remarks.

\begin{appendix}
\section{Solution to the initial value problem for the discrete Volterra chain with $N=3$ }

\begin{center}Yuri Fedorov \footnote{Department of Mathematics, Politechnic University of Catalonia, Barcelona, Spain. E-mail:
{\tt Yuri.Fedorov@upc.edu}} \end{center}
\medskip

Here we provide an alternative construction of explicit sigma-function solutions for the map dVC$_3$ given by eq. (\ref{eq: dVL3}), namely,
\begin{gather}
 \tilde x_1 -x_1 = \varepsilon x_1 (\tilde x_2- \tilde x_3) +\varepsilon \tilde x_1 ( x_2- x_3), \notag \\
  \tilde x_2 -x_2 = \varepsilon x_2 (\tilde x_3- \tilde x_1) +\varepsilon \tilde x_2 ( x_3- x_1), \label{28} \\
\tilde x_3 -x_3 = \varepsilon x_3 (\tilde x_1- \tilde x_2) +\varepsilon \tilde x_3 ( x_1- x_2) . \notag 
\end{gather}
Our method, like the method of the main paper, is based on an extensive use of the symbolic manipulation software Maple.
Recall that the map dVC$_3$ possesses the conserved quantities given by eq. \eqref{eq: dVL3 ints},
\begin{equation} \label{ints}
 H_1(x) =x_1+x_2+x_3, \quad H_2(x,\varepsilon) = 
\frac{x_1 x_2 x_3 }{1-\varepsilon^2 (x_1^2+x_2^2+x_3^2- 2 x_1 x_2 -2 x_2 x_3 - 2 x_3 x_1 ) },
\end{equation}
which implies that the complex invariant manifold $\{x\in {\mathbb C}^3 : H_1(x)=H_1, \, H_2(x)=H_2\}$ of the map is 
the intersection of a cubic surface and a plane, and, therefore, in general, is an open subset of an elliptic curve. For fixed constants of motion $H_1, H_2$, we eliminate $x_3$ from the integrals \eqref{ints} and get the equation of the planar elliptic curve $E$:
$$
f(x_1,x_2)= 1 - \varepsilon^{2}\Big(H_1^2 - 4x_{1}x_{2} - 
4(H_1 - x_1 -x_2)(x_1 + x_2)\Big) - \frac { x_1x_2(H_1 - x_1 - x_2) } {H_2 }  =0 .
$$
Then we apply the classical algorithm of parameterization of elliptic curves embedded in ${\mathbb C}^2$, implemented in the Maple command {\tt Weierstrassform}$(f,x_1,x_2,X,Y)$. It immediately returns a canonical form of this curve:
\begin{equation} \label{XY}
  Y^2 = P_3(X)=\frac{1}{4} (4 X^3- g_2 X- g_3) \,,
\end{equation}
where
\begin {align}
 g_2= &\frac {1}{12}(H_1-12 \varepsilon^2 H_2) (H_1^3 - 24H_2 - 12\,\varepsilon^2 H_1^2 H_2 + 48\,\varepsilon^4 H_1H_2^2-192\,\varepsilon^6 H_2^3  ), \label{g2}\\
 g_3 = &  -H_2^2 + \frac{1}{6}H_1^3H_2 - \frac{1}{216}H_1^6 - 4\,\varepsilon^2H_1^2H_2^2 + \frac{1}{6}\varepsilon^2 H_1^5H_2
                - \frac{7}{3}\varepsilon^4 H_1^4H_2^2 + 40\,\varepsilon^4 H_1H_2^3\nonumber\\
            & -160\,\varepsilon^6 H_2^4 + \frac{40}{3}\varepsilon^6 H_1^3H_2^3 - 256\,\varepsilon^{10} H_1H_2^5 + 512\,\varepsilon^{12} H_2^6. \label{g3}
\end{align}
The same Maple command also gives the original variables $x_1, x_2$ in terms of $X,Y$:
\begin{align}
x_1 & = \frac {H_1^3 - 12\,H_2 - 16\,\varepsilon^2 H_1^2 H_2  + 48\,\varepsilon^4 H_1 H_2^2 +192\,\varepsilon^6 H_2^3
            - 24\,Y -12 B X } {2 (b^2-12 X)} , \label{xx1} \\
x_2 & =  
\frac {4H_2(3 - 2\,\varepsilon^2 H_1^2+ 24\,\varepsilon^4 H_1 H_2 - 48\,\varepsilon^6 H_2^2  )- 48\,\varepsilon^2H_2X } {b^2-12 X} , \label{xx2}
\end{align}
where the following abbreviations are used:
$$
  b= H_1-12 \varepsilon^2 H_2\,   , \quad B=H_1-4 \varepsilon^2 H_2  \,.
$$

Then the relation $x_1+x_2+x_3=H_1$ yields
\begin{equation}\label{xx3}
x_3  = \frac {H_1^3 - 12\,H_2  - 16\,\varepsilon^2 H_1^2 H_2 + 48\,\varepsilon^4 H_1 H_2^2 + 192\,\varepsilon^6 H_2^3
             + 24\,Y - 12 B X} {2 (b^2-12 X)}. 
\end{equation}
 
 The standard parametrization of the Weierstrass canonical form of $E$ is:
\begin{equation}  \label{wpXY} 
X=\wp(u;g_2, g_3), \quad  Y= \frac 1 2 \wp'(u; g_2,g_3), \quad  
u\in {\mathbb C}.
\end{equation} 
To obtain the parameterization of $E$ in terms of sigma-functions of $u$, 
we note that near the only infinite point $\infty$ of $E$ we have the expansions 
\begin{equation} \label{XY_exp}
X= 1/u^2 +O(u), \quad Y=-1/u^3+ O(1).
\end{equation} 
Then the above expressions $x_i(X,Y)$ imply that $x_2(u)$ has  a pair of simple poles at $u=\pm v$, with 
\begin{equation} \label{hh}
\wp(v) = \frac {b^2}{12}, \quad \wp'(v) = -H_2 (1+\varepsilon H_1 - 8 \varepsilon^3 H_2) (1-\varepsilon H_1 + 8 \varepsilon^3 H_2) ,
\end{equation}
and
a pair of simple zeros at $u=\pm \delta$ with 
\begin{equation}  \label{del} 
\begin{aligned}  
\wp(\delta)= & \frac {1}{4 \varepsilon^2} - \frac{1}{6} ( H_1^2-12\,\varepsilon^2 H_1H_2 +24 \varepsilon^2 H_2^2),\\
\wp'(\delta)= & - \frac {1}{4 \varepsilon^3} (1+\varepsilon H_1  - 8 \varepsilon^3 H_2)(1- \varepsilon H_1+ 8 \varepsilon^3 H_2 )  .   
\end{aligned}
\end{equation}
Hence, we have a parametrization in terms of the sigma function $\sigma(u; g_2, g_3)$:
\begin{equation}\label{x2 prelim}
   x_2(u) = \rho_2 \frac{\sigma (u-\delta)\,\sigma (u+\delta)}  {\sigma (u-v)\,\sigma (u+v)}, 
\end{equation}
with a constant factor $\rho_2$, which is calculated from the condition 
$$
x_2(u=0) = \rho_2 \frac{\sigma^2 (\delta)}{\sigma^2(v)} = \lim_{X\to \infty} x_2(X) = 4 \varepsilon^2 H_2.
$$

Next, in view of \eqref{XY_exp} the denominator of $x_1(X,Y)$ in \eqref{xx2} has simple zeros at 
$u=\pm v$ and a double pole at $u=0$. 
The numerator has a triple pole at $u=0$. The zeros of the numerator are defined by the relation
\begin{equation} \label{line1}
  Y= -\frac B2 X + \frac {1}{24} H_1^3 -\frac 12 H_2  - \frac 23 \varepsilon^2 H_1^2 H_2  + 
2 \varepsilon^4 H_1 H_2^2  + 8 \varepsilon^6 H_2^3 .
\end{equation}
Substituting this into the canonical equation \eqref{XY}, we get the following condition for $X$:
$$
 (12 X- b^2) (12 X -H_1^2 -24 \varepsilon H_2 +48 \varepsilon^4 H_2^2 ) (12 X -H_1^2 +24 \varepsilon H_2 +48 \varepsilon^4 H_2^2 ) =0 .
$$
The latter condition in combination with  \eqref{line1} defines 3 points on the curve $E$. 
One concludes that the numerator vanishes at $u = v, \mu_1, \mu_2$ with
\begin{equation} \label{mu's}
\begin{aligned}
 \wp(\mu_1) = \frac{1}{12} H_1^2+2\varepsilon H_2 - 4 \varepsilon^4 H_2^2 , \quad  
 \wp'(\mu_1)= -H_2 (1+\varepsilon H_1)(1+\varepsilon H_1-8\varepsilon^3 H_2), \\
\wp(\mu_2) = \frac{1}{12} H_1^2-2 \varepsilon H_2-4 \varepsilon^4 H_2^2, \quad  
 \wp'(\mu_2)= -H_2 (1-\varepsilon H_1)(1-\varepsilon H_1+8\varepsilon^3 H_2) . 
\end{aligned}
\end{equation}
Combining these observation, we conclude that $x_1$ has simple poles at $u=0, -v$ and simple zeros at $u=\mu_1, \mu_2$ and no zeros or poles elsewhere. As a result, one can write 
\begin{equation}\label{x1 prelim}
   x_1(u) = \rho_1 \frac{\sigma (u-\mu_1)\,\sigma (u-\mu_2)}  {\sigma (u)\,\sigma (u+v)},
\end{equation}
with a suitable constant factor $\rho_1$. The expressions $x_1(X,Y), x_3(X,Y)$ differ only by sign of $Y$, hence $x_3(u)=x_1(-u)$, that is, 
$$ % \begin{equation}\label{x3 prelim}
   x_3(u) = \rho_1 \frac{\sigma (u+\mu_1)\,\sigma (u+\mu_2)}  {\sigma (u)\,\sigma (u-v)}.
$$ % \end{equation}
It follows that $\mu_1+\mu_2 \equiv -v$, where $\equiv$ denotes equality modulo the period lattice of $E$. 
Thus, we arrive at
$$ % \begin{align*} 
  x_1(u) = \rho_1 \frac{\sigma (u-\mu_1)\,\sigma ( u +\mu_1 +v)} {\sigma (u)\,\sigma (u+v)},  \quad 
 x_3(u)  = \rho_1 \frac{\sigma (u -\mu_1 -v)\,\sigma (u+\mu_1)} {\sigma (u)\,\sigma (u-v)}.
$$ % \end{align*}

Since $x_1,x_2,x_3$ enter into the integrals $H_1(x), H_2(x)$ in a symmetric way, there is a cyclic group acting on $E$
by translations: 
\begin{equation} \label{3a}
x_1=f(u), \quad x_2=f(u+h), \quad x_3=f (u+2h), \quad x_1=f (u+3h), \qquad 3 h \equiv 0. 
\end{equation}
The above formulas for $x_i(u)$ are compatible with equation \eqref{3a} if and only if  $h\equiv v$ and $\delta\equiv\pm (v-\mu_1)$. By a straightforward computation, one shows that
$$
\left| \begin{array}{ccc}
\wp(v) & \wp'(v) & 1 \\
\wp(\delta) & -\wp'(\delta) & 1 \\
\wp(\mu_1) & -\wp'(\mu_1) & 1 
\end{array}\right|=0,
$$
which, according to the classical addition theorem for the $\wp$-function, yields $\delta\equiv v- \mu_1$.  As a consequence, we can write
\begin{align} 
 x_1(u) & = \rho_1 \frac{\sigma (u+ \delta-v)\,\sigma ( u-\delta +2v )} {\sigma (u)\,\sigma (u+v)},  \notag \\ 
 x_2(u) & = \rho_2 \frac{\sigma (u+\delta )\,\sigma (u-\delta )} {\sigma (u+v )\,\sigma (u-v)}, \label{x_u} \\ 
 x_3(u) & = \rho_1 \frac{\sigma (u+\delta -2v )\,\sigma(u-\delta+v)} {\sigma (u)\,\sigma (u-v)}. \notag 
\end{align}
By shifting the argument $u$ by the period $3v$ in some terms above, these relations take the form 
\eqref{3a}:
\begin{align} 
  x_1(u) & = \bar\rho \, \frac{\sigma (u-\delta-v)\,\sigma ( u +\delta - v )} {\sigma (u)\,\sigma (u-2v)}, \notag \\ 
  x_2(u) & = \bar \rho\, \frac{\sigma (u- \delta )\,\sigma(u +\delta )} {\sigma (u+v )\,\sigma (u-v)}, \label{x_new} \\ 
 x_3(u) & = \bar\rho\, \frac{\sigma (u-\delta+v)\,\sigma (u +\delta+v)} {\sigma (u+2v)\,\sigma (u)}, \notag 
\end{align}
where $\bar\rho=\rho_2 = 4 \varepsilon^2 H_2 \dfrac{\sigma^2(v)}{\sigma^2(\delta)}$ and, as follows from the above, 
$  v = \int_\infty^{b^2/12} \frac{dX}{2\sqrt{P_3(X)}}$.

Observe that, like in the continuous case, the condition  
$3 v \equiv 0$ can be verified analytically. Namely, using \eqref{hh}, \eqref{g2}, we compute:
$ \wp''(v) = 6\wp^2(v)- g_2/2 = - b\wp'(v)$. Then the classical duplication formula for $\wp(v)$ yields
$$
 \wp(2v) = \frac 14 \left( \frac{\wp''(v)}{ \wp'(v) }\right)^2-2 \wp(v) = \frac{b^2}{4} - 2\wp(v) = 3\wp(v) - 2\wp(v) =\wp(v).    
$$
Since $v \ne 0$, we find: $2v \equiv - v$, so that $3v \equiv 0$.   

 Given the values of the coordinates $x_{1}, x_{2}, x_{3}$,  and therefore 
the corresponding constants $H_1, H_2, \delta, v$, 
the value of the parameter $u$ can be recovered from equation \eqref{xx2}, which yields   
\begin{equation} \label{x0}
  X= \frac{x_{2} \wp(v) - 4 \varepsilon^2 H_2 \wp(\delta)}{x_{2}-4 \varepsilon^2 H_2}\, ,
\end{equation}
so that, up to sign, $u = \int_\infty^{X} \frac{dX'}{2\sqrt{P_3(X')}}$. The sign of $u$ can be determined by verifying, say, the first
expression in \eqref{x_new} with the given $x_1$.

\bigskip

{\bf dVC$_3$ as a translation on $E$.} Since the map \eqref{28} is algebraic, its action on $E$ must be a
translation by a constant vector. We show that the following result holds true.
\begin{theorem} \label{th appendix}

The shift on $E$ generated by the map \eqref{28} is equal to $\delta$, so that the orbits of dVC$_3$ $x(n)=f^n(x(0))$ are given by the expressions
\begin{align} 
  x_{1}(n) & = \bar\rho \, \frac{\sigma (u_n-\delta-v)\,\sigma ( u_n +\delta-v )} {\sigma (u_n)\,\sigma (u_n-2v)}, \notag \\ 
  x_{2}(n) & = \bar \rho\, \frac{\sigma (u_n- \delta )\,\sigma(u_n +\delta )} {\sigma (u_n+v )\,\sigma (u_n-v)}, 
\label{x_n} \\ 
 x_{3}(n) & = \bar\rho\, \frac{\sigma (u_n-\delta+v)\,\sigma (u_n +\delta+v)} {\sigma (u_n+2v)\,\sigma (u_n)}, \notag 
\end{align}
where $u_n = n\delta+u_0$, and $v, \delta$ are specified in \eqref{hh}, \eqref{del}. The phase $u_0$ is determined from the initial values of $x_1,x_2,x_3$ as described above. 
\end{theorem}

\begin{proof}
Substitute the expressions $x_i(X,Y)$, $i=1,2,3$ into  \eqref{28} and obtain, up to a constant factor, 
\begin{equation}  \label{tilde_x}
  \tilde x_1  = x_1 \frac{ (X-\wp(v))\, U(X,Y) } { (X-\wp(\mu_1))  \,(X-\wp(\mu_2))\, (X-\wp(\delta)) },   
\end{equation}
where
\begin{align}  
 U(X,Y) & =  12 (1-\varepsilon H_1 +12\, \varepsilon^3 H_2 +4 \,\varepsilon^4 H_1 H_2 )\,X +
24\,\varepsilon(1-\varepsilon H_1) Y \notag \\ 
& \quad - H_1^2 + \varepsilon H_1^3 - 36\,\varepsilon H_2 + 12\,\varepsilon^2 H_1 H_2 - 144\,\varepsilon^4 H_2^2 + 8\,\varepsilon^4 H_1^3 H_2 
\notag \\
 & \quad - 144\,\varepsilon^5 H_1 H_2^2 - 96\,\varepsilon^6 H_1^2 H_2^2 + 576\,\varepsilon^7 H_2^3 + 192\,\varepsilon^8 H_1 H_2^3. 
\end{align}
In view of equation \eqref{XY}, on $E$ the function $U(X,Y)$ has simple zeros at $u=- \mu_1$, $u=\delta$, and at a point $u=q$ with 
\begin{gather*}
\wp(q) = {\displaystyle \frac {1}{12}} (H_1^2 - 2\,\varepsilon H_1^3 + 48\,\varepsilon H_2+\varepsilon^{2}H_1^4 
   - 24\,\varepsilon^2 H_1 H_2 - 24\,\varepsilon^4 H_1^3 H_2 + 528\,\varepsilon^4 H_2^2  \\
\mbox{} + 96\, \varepsilon^5 H_1 H_2^2 + 144\,\varepsilon^6 H_1^2 H_2^2) \left/ {\vrule height0.54em width0em depth0.54em
} \right. \!  \! (1-\varepsilon H_1)^{2} , \\
 \wp'(q)= -H_2(1 - \varepsilon H_1+8\,\varepsilon^3 H_2)
    (1+2\,\varepsilon H_1 - 2\,\varepsilon^{3} H_1^3 + 72\,\varepsilon^3 H_2 - \varepsilon^4 H_1^4 + 24\,\varepsilon^4 H_1H_2   \\
\mbox{} + 24\,\varepsilon^5 H_1^2 H_2  +8\,\varepsilon^6  H_1^3 H_2)
 \left/ {\vrule height0.54em width0em depth0.54em} \right. \! 
 \! (1-\varepsilon H_1)^3 .
\end{gather*}
On the other hand, $U(X,Y)$ has a triple pole at $u=0$. 
Then \eqref{tilde_x} implies   
that the function $\tilde x_1(u)/x_1(u)$ has a simple zero at $X=\infty$, that is, at $u=0$. Summarizing, we have:
$$
 \tilde x_1 (u)\sim x_1 (u) \frac {\sigma(u) \sigma(u-v) \sigma(u+v) \sigma(u-q) }
{\sigma(u+\delta) \sigma(u-\mu_1) \sigma(u-\mu_2)  \sigma(u+\mu_2) } ,
$$
where ``$\sim$'' means equality up to a constant factor. Recall that $\mu_1+\mu_2\equiv -v$ and $\mu_1\equiv  v-\delta$, so that $\mu_2\equiv\delta-2v$.
In view of  \eqref{x1 prelim},
$$
x_1(u)\sim \frac{\sigma(u-\mu_1)\sigma(u-\mu_2)}{\sigma(u)\sigma(u+v)}=\frac{\sigma(u-v+\delta)\sigma(u+2v-\delta)}{\sigma(u)\sigma(u+v)},
$$ 
we find:
$$
   \tilde x_1(u) \sim \frac {\sigma(u-v) \sigma(u-q) }{\sigma(u+\delta) \sigma(u+\mu_2) }= \frac {\sigma(u-v) \sigma(u-q) }{\sigma(u+\delta) \sigma(u+\delta-2v) }  \sim \frac { \sigma(u+ 2v) \sigma(u-q) }
{\sigma(u+\delta) \sigma(u + v +\delta ) }. 
$$
By comparing the poles of $x_1(u)$ and $\tilde x_1(u)$, we conclude that  the translation on $E$ is given by $u \mapsto u+ \delta$. 
\end{proof}

Expressions \eqref{x_new} and the result of Theorem \ref{th appendix} confirm formulas \eqref{eq: dVC3 solution sigma1}--\eqref{eq: dVC3 solution sigma3} of the main body of the paper, and provide us with the further information, like the expressions of the Weierstrass invariants $g_2$, $g_3$ and of the shift $\delta$ through the integrals of motion $H_1$, $H_2$. Observe that in the continuous limit $\varepsilon \to 0$ the invariants \eqref{g2}, \eqref{g3} tend to \eqref{eq: VC3 g2g3 thru ints}, and similarly expressions \eqref{hh} tend to $\wp(v)= H_1^2/12$, $\wp'(v)=  -H_2$, which coincide with \eqref{eq: VC3 v thru ints}.
Moreover, according to \eqref{del} and the asymptotical behavior of $\wp(u)$ near zero,  for small $\varepsilon$ one has the expansion 
\begin{equation} \label{delta-epsilon}
\delta= 2 \varepsilon + O(\varepsilon^2) \, .   
\end{equation}
\smallskip

{\bf Real orbits of dVC$_3$.} 
Let $\omega_1, \omega_3$ be half-periods of the elliptic curve $E\in {\mathbb C}^3$, and $\omega_1\in {\mathbb R}$.
Depending on the values $H_1, H_2$, the polynomial $P_3(X)$ in \eqref{XY} 
can have one or 3 real roots.
Then the real part $E_R$ of $E$ has one, respectively two connected components. In the parallelogram of periods they 
are given by segments parallel to the
real axis in the complex plane $u$. In the second case one can choose the second half-period $\omega_3$ to be
imaginary, and the two segments are $\{ 2\omega_1 t: t\in [0,1] \}$ and $\{ 2\omega_1 t +\omega_3: t\in [0,1]\}$. Observe that, according to \eqref{hh}, the 1/3-period $v$ is always real.        

Next, as follows from \eqref{delta-epsilon}, 
for sufficiently small $\varepsilon$, the shift $\delta$ is also small and real, hence the dVC$_3$-orbit of an
initial point on $E_R$ belongs to the same connected component of $E_R$. The situation is different for 
relatively big $\varepsilon$, when $\delta$ becomes a sum of a real number and the imaginary half-period $\omega_3$. Then 
the point $(x_{1}, x_{2}, x_{3})$ changes the component of $E_R$ under each iteration of the map.  

For some special values of $H_1, H_2,\varepsilon$ satisfying $\varepsilon (H_1-8\varepsilon^2 H_2)=\pm 1$ the polynomial 
$P_3(X)$ has a {\it double} root which coincides with the expression for $\wp(\delta)$ in \eqref{del}. 
In this case the invariant curve $\{x\in {\mathbb C}^3 : H_1(x)=H_1,\, H_2(x)=H_2\}$ becomes rational and, as one can show,
% so $\delta$ is a half-period of $E$, and 
the orbit $\cO(x) $ consists of just two points on it. 
\bigskip

{\bf Numerical examples.} Choosing the initial point $(x_1, x_2,x_3)=(3,4,5)$ and $\varepsilon=1$, we obtain $H_1=12$, $H_2=4/3$, and 
the curve $E$ in the Weierstrass form is 
$$
   Y^2 = X^3 -\frac{88}{27} X+ \frac{1636}{729} \quad \text{with} \quad g_2= \frac{352}{27}, \; g_3=-\frac{6544}{729}. 
$$
Up to $10^{-8}$, the $X$-coordinates of the branch points are
$$
    e_1 = -2.0825183, \quad e_2 = 0.9600348, \quad e_3= 1.12248349\, ,
$$
and the real and imaginary half-periods of $E$ are, respectively, 
 $\omega_1= 1.62107698$ and $\omega_3=0.88886315\,i$. The real part $E_R \subset {\mathbb R}^3$ 
consists of two connected components, one is compact, whereas the other one is not.   

Next, following \eqref{hh}, \eqref{del}, one has $\wp(v)=4/3$, $\wp(\delta)=41/36$, and the Abel map gives 
$$
v= \int_{\infty}^{\wp(v)} \frac{dX}{Y}= 1.08084313, \quad \delta= \int_{\infty}^{\wp(\delta)} \frac{dX}{Y}=1.44663208.
$$
Observe that $v=2\omega_1/3$  (up to $10^{-7}$). 
%Similarly, from \eqref{mu's} one finds 
%$$
%\wp(\mu)=\wp(\mu_1)= -68/9, \quad \wp(\mu_2)=-20/9,
%$$ 
%and then 
%$\mu= -0.366595, \, \mu_2= - 0.715599$.  
Further, the constant $\bar\rho$ in 
\eqref{x_new} equals $\bar\rho = 3.71953594$ and, according to \eqref{x0}, the initial phase equals 
$u_0= 1.265142+ \omega_3$.  Then the evaluation of expressions
\eqref{x_new} with $u=u_0$  with Maple function {\tt WeierstrassSigma$(u,g_2, g_3)$} recovers the initial values $(x_1, x_2, x_3)=(3,4,5)$ up to $10^{-6}$.  

Finally, the first iteration of the map \eqref{28} yields 
$$
x_1(1)=71/15, \quad x_2(1)=212/45, \quad x_3(1)=23/9,
$$ 
and the formulas \eqref{x_n} for $n=1$ give the same values up to $10^{-4}$.       
\medskip

In the second example we take the initial point $(x_1,x_2,x_3)=(3,-4,5)$, and again $\varepsilon=1$, which give $H_1=4$, $H_2=60/83$, and 
the curve $E$ in the canonical form reads
$$
   Y^2 = X^3 -\frac{892144696}{142374963} X + \frac{52898505488012}{8827390080963}.
 $$
The branch points ordered according to $e_1 < e_2 <e_3$ are
$$  
e_1= - 2.8880922, \quad e_2=1.342284391, \quad e_3= 1.54580078.   
$$
As in the previous example, $E_R$ consists of two connected components. 
Following \eqref{hh}, \eqref{del}, one finds $\wp(v)=1.882107$, $\wp(\delta) = 1.27617$.  
Without further calculations, we observe that, whereas $\wp(v) > e_3$ and therefore $v\in\mathbb R$, one has
$\wp(\delta) \in (e_1, e_2)$, hence 
the shift $\delta$ is the sum of a real number and the imaginary half-period $\omega_3$. As mentioned above,
in this case an initial point and its $f$-image belong to different components of $E_R$.     

\subsection*{Acknowledgments} The author acknowledges support of the Spanish MINECO-FEDER Grant MTM2012-37070. He is also grateful to Yu. Suris for useful discussions. 
\end{appendix}

%%%%%%%%%%%%%%%%%%%%%%%%%%%%%%%
%%%%%%%%%%%%%%%%%%%%%%%%%%%%%%%

\end{document}